\documentclass[runningheads]{llncs}

\usepackage{marvosym}
\usepackage{wrapfig}
\usepackage{stmaryrd} 
\usepackage{proof}
\usepackage{amsmath}
\usepackage{amssymb}
\usepackage{mathtools}
\usepackage{mathpartir}
\usepackage{color}
\usepackage{tikz}
\usepackage{xstring}
\usepackage{ntabbing}
\usepackage{listings}
\usepackage{varwidth}
\usepackage{enumitem}
\usepackage{multicol}
\usepackage{lipsum}
\usepackage{booktabs}
\usepackage{fancyvrb}
\usepackage{hyperref}

\newcommand{\polymorphicsessiontypes}[0]{nested session types}

\newcommand{\PolymorphicSessionTypes}[0]{Nested Polymorphic Session Types}

\newlength{\rWidth}

\newcommand{\m}[1]{\mathsf{#1}}
\newcommand{\mb}[1]{\mathbf{#1}}

\newcommand{\Sg}{\Sigma}

\newcommand{\G}{\Gamma}

\newcommand{\lolli}{\multimap}
\newcommand{\tensor}{\otimes}
\newcommand{\with}{\mathbin{\binampersand}}

\newcommand{\one}{\mathbf{1}}

\newcommand{\semi}{\; ; \;}
\newcommand{\ichoiceop}{\oplus}
\newcommand{\echoiceop}{\with}
\newcommand{\ichoice}[1]{\mathord{\ichoiceop} \{ #1 \}}
\newcommand{\echoice}[1]{\mathord{\echoiceop} \{ #1 \}}

\newcommand{\mi}[1]{\mbox{\it #1}}

\newcommand{\tforall}[1]{\forall #1. \, }
\newcommand{\texists}[1]{\exists #1. \, }

\newcommand{\Next}{\raisebox{0.3ex}{$\scriptstyle\bigcirc$}}
\renewcommand{\next}[1]{\Next #1}
\newcommand{\tdelay}[2]{
    \IfEqCase{#2}{%
        {1}{\next{#1}}%
    }[{\Next^{#2} (#1)}]%
}%

\newcommand{\vars}{\mathcal{V}}

\newcommand{\unfldsg}[2]{\m{unfold}_{#1}(#2)}
\newcommand{\unfold}[1]{\unfldsg{\Sg_0}{#1}}
\newcommand{\rel}{\mathcal{R}}

\newcommand{\SD}{\mathcal{S}(\mathcal{D})}
\newcommand{\SDO}{\mathcal{S}(\mathcal{D}_0)}

\definecolor{verylightgray}{rgb}{.97,.97,.97}

\lstdefinelanguage{Solidity}{
	keywords=[1]{anonymous, assembly, assert, balance, break, call, callcode, case, catch, class, constant, continue, constructor, contract, debugger, default, delegatecall, delete, do, else, emit, event, experimental, export, external, false, finally, for, function, gas, if, implements, import, in, indexed, instanceof, interface, internal, is, length, library, log0, log1, log2, log3, log4, memory, modifier, new, payable, pragma, private, protected, public, pure, push, require, return, returns, revert, selfdestruct, send, solidity, storage, struct, suicide, super, switch, then, this, throw, transfer, true, try, typeof, using, value, view, while, with, addmod, ecrecover, keccak256, mulmod, ripemd160, sha256, sha3}, 
	keywordstyle=[1]\color{blue}\bfseries,
	keywords=[2]{address, bool, byte, bytes, bytes1, bytes2, bytes3, bytes4, bytes5, bytes6, bytes7, bytes8, bytes9, bytes10, bytes11, bytes12, bytes13, bytes14, bytes15, bytes16, bytes17, bytes18, bytes19, bytes20, bytes21, bytes22, bytes23, bytes24, bytes25, bytes26, bytes27, bytes28, bytes29, bytes30, bytes31, bytes32, enum, int, int8, int16, int24, int32, int40, int48, int56, int64, int72, int80, int88, int96, int104, int112, int120, int128, int136, int144, int152, int160, int168, int176, int184, int192, int200, int208, int216, int224, int232, int240, int248, int256, mapping, string, uint, uint8, uint16, uint24, uint32, uint40, uint48, uint56, uint64, uint72, uint80, uint88, uint96, uint104, uint112, uint120, uint128, uint136, uint144, uint152, uint160, uint168, uint176, uint184, uint192, uint200, uint208, uint216, uint224, uint232, uint240, uint248, uint256, var, void, ether, finney, szabo, wei, days, hours, minutes, seconds, weeks, years},	
	keywordstyle=[2]\color{teal}\bfseries,
	keywords=[3]{block, blockhash, coinbase, difficulty, gaslimit, number, timestamp, msg, data, gas, sender, sig, value, now, tx, gasprice, origin},	
	keywordstyle=[3]\color{violet}\bfseries,
	identifierstyle=\color{black},
	sensitive=false,
	comment=[l]{//},
	morecomment=[s]{/*}{*/},
	commentstyle=\color{gray}\ttfamily,
	stringstyle=\color{red}\ttfamily,
	morestring=[b]',
	morestring=[b]"
}

\newcommand{\ltrans}[1]{\xrightarrow{#1}}

\newcommand{\clo}[2]{\langle #1 \semi #2\rangle}

\newcommand{\cov}{+}
\newcommand{\conv}{-}
\newcommand{\biv}{\top}
\newcommand{\nonv}{\bot}
\newcommand{\ofv}{\mathrel{\mbox{\tt\#}}}

\makeatletter
    \let\@internalcite\cite
    \def\cite{\def\citeauthoryear##1##2{##1, ##2}\@internalcite}
    \def\shortcite{\def\citeauthoryear##1{##2}\@internalcite}
    \def\@biblabel#1{\def\citeauthoryear##1##2{##1, ##2}[#1]\hfill}
\makeatother

\begin{document}

\title{Subtyping on Nested Polymorphic Session Types}

%
\author{Ankush Das\inst{1} \and
Henry DeYoung\inst{1} \and
Andreia Mordido\inst{2} \and 
Frank Pfenning\inst{1}}
\authorrunning{A. Das et al.}
\institute{Carnegie Mellon University, USA \and
LASIGE, Faculdade de Ci\^encias, Universidade de Lisboa, Portugal}

\maketitle

\begin{abstract}
The importance of subtyping to enable a wider range of well-typed programs is undeniable. However, the interaction between subtyping, recursion, and polymorphism is not completely understood yet.  In this work, we explore subtyping in a system of nested, recursive, and polymorphic types with a coinductive interpretation, and we prove that this problem is undecidable. Our results will be broadly applicable, but to keep our study grounded in a concrete setting, we work with an extension of session types with explicit polymorphism, parametric type constructors, and nested types. We prove that subtyping is undecidable even for the fragment with only internal choices and nested unary recursive type constructors. Despite this negative result, we present a subtyping algorithm for our system and prove its soundness. We minimize the impact of the inescapable incompleteness by enabling the programmer to seed the algorithm with subtyping declarations (that are validated by the algorithm). We have implemented the proposed algorithm in Rast and it showed to be efficient in various example programs.
\end{abstract}

\section{Introduction}
\label{sec:intro}

Subtyping is a standard feature in modern programming languages,
whether functional, imperative, or object-oriented.  The two principal
approaches to understanding the meaning of subtyping is via coercions
or as subsets.  Either way, subtyping allows more programs as
well-typed, programs can be more concise, and types can express more
informative properties of programs.

When it comes to the interaction between advanced features of type
systems, specifically recursive and polymorphic types, there are still
some gaps in our understanding of subtyping.  In this paper we analyze
a particular interaction of features, but as we explain below, we
contend that the lessons are broadly applicable.
\begin{itemize}
\item Recursive types and type constructors are \emph{equirecursive}
  (rather than isorecursive) and \emph{structural} (rather than
  generative).
\item Recursively defined type constructors may be \emph{nested}
  rather than being restricted to be regular (in the terminology
  of \cite{Solomon78popl}).
\item Types are interpreted \emph{coinductively} rather than
  inductively.
\item Polymorphism is \emph{explicit} rather than implicit.
\item Types are \emph{linear}.
\end{itemize}
We prove that subtyping for a language with even a small subset of
these features is \emph{undecidable}, including only internal choice
and nested unary recursive type constructors, but no function or
product types, type quantification, or type constructors with multiple
parameters.  This may be disheartening, but we also present an
incomplete algorithm that has performed well over a range of examples.
One of its features is that it is inherently extensible through
natural programmer declarations.  We provide a simple prototypical
example of how nested recursive types and subtyping can cooperate to
allow some interesting program properties to be checked and expressed.
Other examples of nested polymorphic types and applications in
functional programming can be found in the
literature~\cite{Bird98mpc,Okasaki98book,Hinze00jfp,Johann09hosc}.
Closely related issues arise in object-oriented programs with
path-dependent types~\cite{Amin16wadler,Mackay20popl}.  Recent
semantic investigations have demonstrated that nested types exhibit
many of the properties we come to expect from regular types, including
induction~\cite{Johann20fossacs} and parametricity
principles~\cite{Johann21arxiv}.

This paper continues an investigation of nested equirecursive
types~\cite{Das21esop}.  This prior work did not treat quantified
types and considered only \emph{type equality} rather than subtyping.
Remarkably, type equality is \emph{decidable}, although no practical
and complete algorithm is known.  In~\cite{Das21esop}, we also proposed a
practical, but incomplete algorithm, which provided some inspiration
for the algorithm here.  In experiments with nested recursive types
restricted to type equality we found that some programs became
prohibitively complex, since we needed to explicitly implement (often
recursive) coercions between subtypes.  This provided some motivation
for the generalizations present in this paper.  We found that the
additional complexities of subtyping are considerable, regarding both
pragmatics and theory.

Which of our language features above are essential for our
undecidability result and practical algorithm?  The nature of the
subtyping changes drastically if recursive types are \emph{generative}
(or \emph{nominal}), essentially blocking many applications of
subtyping.  However, the issues addressed in this paper arise again if
datasort refinements~\cite{Freeman91,Dunfield04popl,Davies05phd} are
considered.  Specifically, the interaction between polymorphism and
datasort refinements was identified as problematic~\cite{Skalka97ms}
and somewhat drastically restricted in SML
Cidre~\cite{Davies11github}.  Our results therefore apply to
refinement types in the presence of polymorphic constructors.

Regarding nesting, in the non-generative setting it seems difficult to
justify the exclusion of nesting.  If we disallow it and artificially
restrict subtyping, then it becomes decidable by standard subtyping
algorithms, whether types are interpreted
inductively~\cite{Ligatti17toplas} or
coinductively~\cite{Amadio93toplas,Brandt98fi,Gay05acta} with the standard notion
of \emph{constructor variance}~\cite{Igarashi02ecoop,Kennedy06tr,Altidor11pldi}.  
We also expect that, once nesting is permitted, there is no essential
difference between equi- and isorecursive subtyping.  We expect that it
will be easy to map isorecursive types to corresponding equirecursive
types while preserving subtyping, in which case our algorithm could be 
applied to isorecursive types.  But we leave the details to future work.

Our system of subtyping remains sound when certain types (for example,
purely positive recursive types) are considered inductively.  However,
it then suffers from an additional source of incompleteness by not
identifying all empty types~\cite{Ligatti17toplas}.  Since in
particular function types are almost always considered nonempty (in
essence, coinductively) this does not change our undecidability
result.  For a purely inductive fragment of the types this can also be
obtained directly by a reduction from the undecidability of the
language inclusion problem for deterministic pushdown
automata~\cite{Friedman76tcs} following
Solmon's construction~\cite{Solomon78popl}.

Linearity actually plays no role whatsoever in subtyping, and we use
it here only because coinductively defined linear types naturally
model binary session
types~\cite{Honda93concur,Gay10jfp,Caires10concur,Wadler12icfp}.
Session typed concurrent programs form the basis for our experimental
implementation in Rast~\cite{Das20fscd}.

However, because our type simulation is based on the observable
communication behavior of processes (like \cite{Gay05acta}), the choice of Rast as an underlying programming language is of secondary
interest.  We chose Rast
because its coinductive interpretation of types and bidirectional
type checker allow for easy experimentation.
Lastly, type safety for Rast with explicit polymorphism can be proved in a straightforward way
following our previous work~\cite{Das21esop}.

In summary, our main contributions are:
\begin{itemize}
	\item A definition of subtyping for session types with parametric 
	type constructors, explicit polymorphism, and nested types.
	\item A proof of undecidability of subtyping for even a small fragment 
	of the proposed nested polymorphic session types.
	\item A practical (but incomplete) algorithm to check subtyping and 
	its soundness.
	\item A sound extension of the subtyping algorithm through 
	programmer declarations, to minimize the impact of incompleteness.
	\item An implementation of the subtyping algorithm.
\end{itemize}

The remainder of the paper is organized as follows.  We provide an
overview and simple examples in Section~\ref{sec:overview}.  Our
language of types and criteria for their validity are presented in
Section~\ref{sec:types}.  Section~\ref{sec:subtyping} gives a semantic
definition of subtyping followed by an undecidability proof by
reduction from language inclusion for Basic Process Algebras.  In
Section~\ref{sec:subtp-alg}, we present our algorithm for
subtyping and its soundness proof.  Section~\ref{sec:implementation}
briefly sketches the implementation, followed by further examples in
Section~\ref{sec:examples}.  We make some comparisons with related work in
Section~\ref{sec:related-work} and then conclude with
Section~\ref{sec:conclusion}.

\section{Overview}
\label{sec:overview}

Understanding the interaction between polymorphism and subtyping is not a new  
concern. In this paper, we study the interaction of explicit polymorphism,
recursion, type nesting, and subtyping on (linear) session types~\cite{Honda93concur,Gay10jfp,Caires10concur}. 
However, nothing in our analysis 
depends on linearity, so our key results apply more broadly to any system of 
equirecursive, structural, coinductive types with nesting and
explicit polymorphism.

Consider a type of unary representations of natural numbers given by 
\begin{equation*}
  \mathsf{nat} = \ichoice{\mathbf{z} : \one, \mathbf{s} : \mathsf{nat}}
  \,.
\end{equation*}
Each $\mathsf{nat}$ is either $\mathbf{z}$ or $\mathbf{s}$ followed by a $\mathsf{nat}$.
As a session type, it is represented as an \emph{internal choice}
between two possible labels: $\mathbf{z}$ terminates the communication (represented by $\one$), 
and $\mathbf{s}$ recursively calls $\mathsf{nat}$.
The types of even and odd natural numbers can now be defined as
\begin{equation*}
  \mathsf{even} = \ichoice{\mathbf{z} : \one, \mathbf{s} : \mathsf{odd}}
  \quad\text{and}\quad
  \mathsf{odd} = \ichoice{\mathbf{s} : \mathsf{even}}
  \,.
\end{equation*}
As mentioned in the previous section, our recursive types and type constructors are equirecursive and structural, not generative.
As such, it becomes unavoidable that we ought to have $\mathsf{even}$ and $\mathsf{odd}$ as subtypes of type $\mathsf{nat}$, based on a subset interpretation of subtyping.

Moreover, type nesting is also unavoidable in the structural setting.
Consider two type definitions for lists that differ only in their type constructors' names:
\begin{equation*}
  \mathsf{List}[\alpha] = \ichoice{\mathbf{nil} : \one, \mathbf{cons} : \alpha \tensor \mathsf{List}[\alpha]} \text{ and }
  \mathsf{List}'[\alpha] = \ichoice{\mathbf{nil} : \one, \mathbf{cons} : \alpha \tensor \mathsf{List}'[\alpha]}
  \,.
\end{equation*}
Each $\mathsf{List}[\alpha]$ is either $\mathbf{nil}$ or a $\mathbf{cons}$ followed by an $\alpha$ and another $\mathsf{List}[\alpha]$.
As a session type, it is an internal choice between sending either the $\mathbf{nil}$ or $\mathbf{cons}$ labels.
If $\mathbf{cons}$ is chosen, then the a channel of type $\alpha$ is sent, after which the session continues at type $\mathsf{List}[\alpha]$.
In a generative setting, the types $\mathsf{List}[A]$ and $\mathsf{List}'[B]$ would necessarily be unrelated by subtyping -- generativity would ensure that $\mathsf{List}[-]$ would remain distinct from $\mathsf{List}'[-]$.
However, in the structural setting, $\mathsf{List}[A]$ ought to be a subtype of $\mathsf{List}'[B]$ whenever $A$ is a subtype of $B$, again based on a subset interpretation of subtyping -- the two types differ only in the type constructor name, after all.
Moreover, $\mathsf{List}[\mathsf{List}[\mathsf{even}]]$ ought to be a subtype of $\mathsf{List}'[\mathsf{List}'[\mathsf{nat}]]$, and so deciding subtyping relationships between such nested types is inescapable.

Using nested types, we can define type constructors
\begin{equation*}
  \mathsf{Nil} = \ichoice{\mathbf{nil} : \one}
  \quad\text{and}\quad
  \mathsf{Cons}[\alpha,\kappa]= \ichoice{\mathbf{cons} : \alpha \tensor \kappa}
  \,.
\end{equation*}
Both $\mathsf{Nil}$ and $\mathsf{Cons}[A,L]$ are subtypes of $\mathsf{List}[A]$ if $L$ is a subtype of $\mathsf{List}[A]$.
The subtype $\mathsf{Nil}$ characterizes the empty list, whereas $\mathsf{Cons}[A,\mathsf{Cons}[A,\mathsf{Nil}]]$ characterizes lists of length $2$, for example.  (Roughly speaking, $\mathsf{Cons}^n[A,\mathsf{Nil}]$ characterizes lists of length $n$.)

Variances will be inferred for all parameters used in type constructors.
For example, in the type constructor $\mathsf{List}[\alpha]$ above, the type parameter $\alpha$ is covariant because it occurs in only postive positions.
On the other hand, the type $\mathsf{List}[\beta] \lolli \mathsf{List}[\alpha]$ is covariant in $\alpha$ but contravariant in $\beta$ because $\alpha$ and $\beta$ appear in positive and negative positions, respectively.
This means that, for example, $\mathsf{List}[\mathsf{nat}] \lolli \mathsf{List}[\mathsf{nat}]$ is a subtype of $\mathsf{List}[\mathsf{even}] \lolli \mathsf{List}[\mathsf{nat}]$ -- the former type is sufficient anywhere that the latter is required.

As a further example of variances, the type constructor $\mathsf{Seg}[\alpha]$ for list segments given by
\begin{equation*}
  \mathsf{Seg}[\alpha] = \mathsf{List}[\alpha] \lolli \mathsf{List}[\alpha]
\end{equation*}
has both positive and negative occurrences of $\alpha$, owing to the presence of $\lolli$. 
This makes $\alpha$ bivariant in $\mathsf{Seg}[\alpha]$, and it means that, for example,  there is no subtyping relationship between $\mathsf{Seg}[\mathsf{nat}]$ and $\mathsf{Seg}[\mathsf{even}]$ despite $\mathsf{even}$ being a subtype of $\mathsf{nat}$.

Lastly, our language of types includes explicit polymorphic quantifiers $\exists x.\,A$ and $\forall x.\,A$.
We can use the existential quantifier to describe heterogeneous lists:
\begin{equation*}
  \mathsf{HList} = \ichoice{\mathbf{nil} : \one , \mathbf{cons} : \exists x.\, (x \tensor \mathsf{HList})} \,.
\end{equation*} 
Like each $\mathsf{List}[\alpha]$, each $\mathsf{HList}$ is either $\mathbf{nil}$ or a $\mathbf{cons}$; the difference from the type $\mathsf{List}[\alpha]$ is in the existential quantification over the type of data in each $\mathbf{cons}$.
As a session type, $\mathsf{HList}$ types a process that sends either $\mathbf{nil}$ or $\mathbf{cons}$.
If $\mathbf{cons}$ is sent, then some type $T$ and a channel of type $T$ are sent, after which the session continues at type $\mathsf{HList}$ (for the list's tail).
We can even define type constructors
\begin{equation*}
  \mathsf{HNil} = \ichoice{\mathbf{nil} : \one}
  \quad\text{and}\quad
  \mathsf{HCons}[\kappa]= \ichoice{\mathbf{cons} : \exists x.\, (x \tensor \kappa)}
  \,.
\end{equation*}
Our notion of subtyping is such that both $\mathsf{HNil}$ and $\mathsf{HCons}[L]$ are subtypes of $\mathsf{HList}$ if $L$ is a subtype of $\mathsf{HList}$.
However, the explicit polymorphic quantifiers induce communication, so we do not have $\mathsf{Cons}[A,L]$ as a subtype of $\mathsf{HList}$ for any type $A$, even if $L$ is a subtype of $\mathsf{HList}$.

\section{Description of Types}
\label{sec:types}

The underlying base system of session types is derived from a Curry-Howard
interpretation~\cite{Caires10concur,Caires16mscs} of intuitionistic linear logic
\cite{Girard87tapsoft}. Das et al.\ extended this type system with 
parametric type constructors~\cite{Das21esop}. The
\emph{nested polymorphic session types} we propose are also 
endowed with explicit polymorphism, variances, and subtyping.

\subsection{Syntax}
We describe \emph{nested polymorphic session types}, their operational interpretation
and continuation types in Figure~\ref{fig:types}.
\begin{figure}
	\[\arraycolsep=4.5pt
  \begin{array}{lclll}
    A,B,C & ::= & \ichoice{\ell : A_\ell}_{\ell \in L} & \mbox{send label $k \in L$} & \mbox{continue at type $A_k$} \\
          & \mid & \echoice{\ell : A_\ell}_{\ell \in L} & \mbox{receive label $k \in L$} & \mbox{continue at type $A_k$} \\
          & \mid & A \tensor B & \mbox{send channel $a : A$} & \mbox{continue at type $B$} \\
          & \mid & A \lolli B & \mbox{receive channel $a : A$} & \mbox{continue at type $B$} \\
          & \mid & \one & \mbox{send $\m{close}$ message} & \mbox{no continuation} \\
          & \mid & \texists{x} A & \mbox{send type $B$} & \mbox{continue at type $A[B/x]$} \\
          & \mid & \tforall{x} A & \mbox{receive type $B$} & \mbox{continue at type $A[B/x]$} \\
          & \mid & x & \mbox{quantified variable} \\
          & \mid & \alpha & \mbox{type parameter} \\
          & \mid & V[\theta] & \mbox{defined type name} \\
  \end{array}
\]
\vspace*{-1em}
\caption{Description of nested polymorphic session types, their operational
semantics and continuation types.}
\vspace*{-1em}
\label{fig:types}
\end{figure}

The basic type operators have the usual interpretation:
the \emph{internal choice} operator $\ichoice{\ell \colon A_{\ell}}_{\ell\in L}$ 
selects a branch with label $k \in L$ with corresponding
continuation type $A_k$; the \emph{external choice} operator 
$\echoice{\ell \colon A_{\ell}}_{\ell\in L}$ offers a choice
with labels $\ell \in L$ with corresponding continuation types
$A_{\ell}$; the \emph{tensor} operator $ A \tensor B$ 
represents the channel passing type that consists of sending a
channel of type $A$ and proceeding with type $B$;
dually, the \emph{lolli} operator $A \lolli B$ consists of 
receiving a channel of type $A$ and continuing with 
type $B$; the \emph{terminated session} $\one$ is
the operator that closes the session. 

We also have two explicit \emph{type quantifier operators}.
The existential type $\texists{x} A$ is interpreted as sending an arbitrary
well-formed type $B$ and continuing as type $A[B/x]$. Dually, the universal
type $\tforall{x} A$ receives a type $B$ and continues at type $A[B/x]$.

Nested polymorphic session types support \emph{parametrized type definitions} to
define new type names. In a type definition, each \emph{type parameter} is 
assigned a variance. Substitutions $\theta$ for these parameters need
to comply with the prescribed variances. Below we describe variances, 
substitutions, and signatures.
\[
  \begin{array}{llcl}
    \mbox{Variables} & \vars & ::= & \cdot \mid \vars, \alpha \mid \vars, x \\
    \mbox{Variance} & \xi & ::= & \cov \mid \conv \mid \biv \mid \nonv \\
    \mbox{Variances} & \Xi & ::= & \cdot \mid \Xi, \alpha \ofv \xi \\
    \mbox{Substitution} & \theta & ::= & \cdot \mid \theta, A/\alpha \\
    &\sigma & ::= & \cdot \mid \sigma, A/x \\
    \mbox{Signature} & \Sg & ::= & \cdot \mid \Sg, V[\Xi] = A \smallskip \\
  \end{array}
\]

We distinguish \emph{quantified variables} $x$ from \emph{type parameters}
$\alpha$ and, in general, refer to them as \emph{variables}. A \emph{type name} 
$V$ is defined according to a \emph{type definition} $V[\Xi]=A$ in 
\emph{signature} $\Sg$, that is parametrized  
by a sequence of distinct \emph{type parameters} $\overline\alpha$ that the
type $A$ can refer to. Each type parameter $\alpha$ in $\Xi$ has variance $\xi$, 
establishing the position of the occurrences of 
$\alpha$ in $A$: covariant ($\cov$), contravariant ($\conv$), 
bivariant ($\biv$), or nonvariant ($\nonv$). 
We instantiate a definition $V[\Xi] = A$ by writing $V[\theta]$,
where $\theta$ is a substitution for the type parameters in $\Xi$. 
We distinguish substitutions for type parameters, $\theta$, from
substitutions for quantified variables, $\sigma$, because the former
require the validation of variances. 
The set of \emph{free variables} in type $A$ refer to type variables
that occur freely in $A$.
Types without any free variables are called \emph{closed types}.
Any type not of the form $V[\theta]$ is called \emph{structural}.

\subsection{Variances}


We define an implication relation on variances. The implication relation 
constitutes a partial ordering and is defined by the following rules:
\[
  \begin{array}{c}
    \infer[\m{refl}]
    {\xi \leq \xi}{}
    \qquad
    \infer[\bot]
    {\nonv \leq \xi}{}
    \qquad
  \infer[\top]
  {\xi \leq \biv }
  {}
 \end{array}
\]

The least upper bound of this lattice is $\biv$---if a type name $V$ is covariant, 
contravariant, or nonvariant in a type parameter $\alpha$, then it 
also bivariant---whereas
the greatest lower bound is $\nonv$.
The relation $\leq$ on variances can be shown to be transitive.

\begin{wrapfigure}{r}{0.3\textwidth}
\qquad
$ 
  \begin{array}{c||c|c|c|c}
   \xi \mid \xi' & \nonv & \cov & \conv & \biv \\ \hline\hline
  \nonv  & \nonv & \nonv  & \nonv  & \nonv \\ \hline
  \cov & \nonv & \cov & \conv & \biv \\ \hline
  \conv & \nonv & \conv & \cov & \biv \\ \hline
  \biv  & \nonv & \biv  & \biv  & \biv
  \end{array}
$
\end{wrapfigure}

The \emph{nesting} of types is echoed in a \emph{nesting operator} on variances, as defined in  the adjacent table.
Observe  that: nesting with $\cov$ preserves  
variance; nesting with $\conv$ converts covariance to contravariance
and vice versa; nesting with $\biv$ converts covariance and
contravariance to bivariance; and $\nonv$ is the absorbing element 
of this operator. The nesting operator can be naturally extended to a set of variances
as follows:
\begin{mathpar}
	    (\cdot) \mid \xi  = (\cdot) \and
    (\Xi, \alpha \ofv \xi) \mid \xi'  = (\Xi \mid \xi'), \alpha \ofv (\xi \mid \xi')
\end{mathpar}

To simplify notation, we sometimes use $\lnot \xi$ as an abbreviation for $\conv \mid \xi$.

\begin{lemma}[Properties of nesting]\
\begin{itemize}
\item Commutativity: $\xi \mid \xi' = \xi' \mid \xi$ 
\item Associativity: $(\xi_1 \mid \xi_2) \mid \xi_3 = \xi_1 \mid (\xi_2 \mid \xi_3)$
\item Monotonicity: If $\xi \leq \xi'$ then $\xi \mid \zeta \leq \xi' \mid \zeta$
  and if $\zeta \leq \zeta'$ then $\xi \mid \zeta \leq \xi \mid \zeta'$
\end{itemize}	
\end{lemma}
\begin{proof}
	Proofs by case analysis on the variances, 
	using the definition of the nesting operator. Monotonicity also uses
	the definition of the partial order $\leq$.
\end{proof}

\subsection{Signatures, types, and substitutions}
\label{subsec:signatures_types_subst}

Now we present the criteria for valid signatures,
types, and substitutions. 

\paragraph{Valid signatures}
The judgment for valid signatures
is written as 
 $\vdash \Sg\; \mi{valid}$ 
 and is defined by the following rules:
\[
  \begin{array}{c}
    \infer[\m{sig}]
    {\vdash \Sigma_0 \; \mi{valid}}
    {\vdash_{\Sigma_0} \Sigma_0\; \mi{valid}}
    \qquad
    \infer[(\cdot)]
    {\vdash_{\Sigma_0} (\cdot)\; \mi{valid}}
    {\mathstrut}
    \\\\
    \infer[\m{def}]
    {\vdash_{\Sigma_0} (\Sigma, V[\Xi] = A)\; \mi{valid}}
    {\vdash_{\Sigma_0} \Sigma\; \mi{valid} & \cdot \semi \Xi \vdash_{\Sigma_0} A \ofv \cov
    & A \neq V'[\theta]}
  \end{array}
\]

Type definitions may be mutually recursive and each parameter of
each defined type constructor is assigned a variance that is stored in $\Xi$. 
In a \emph{valid signature}, all definitions $V[\Xi] = A$ are
defined over \emph{structural} types 
that are \emph{valid} in a context of variance $\cov$.
(The former restriction is
usually called \emph{contractivity}~\cite{Gay05acta}.)

We take an \emph{equirecursive} view of type definitions, which means
that unfolding a type definition does not require communication.
More concretely, the type $V [\theta]$ is considered equal to its unfolding
$A[\theta]$.
(We expect that we can easily adapt our definitions to an \emph{isorecursive}
view~\cite{Lindley16icfp,Derakhshan19corr} with explicit unfold messages, but we leave the details as future work.)
All type names $V$ occurring in a valid signature
must be defined, and all type parameters defined in a valid
type definition must be distinct.
Furthermore, for a valid definition $V[\Xi] = A$, the free variables occurring
in $A$ must be contained in $\Xi$.

\paragraph{Valid types}
The judgment for valid types
is written as 
$\vars \semi \Xi \vdash_{\Sigma_0} A \ofv \xi$ and 
is defined over a valid signature $\Sg_0$.
$\vars$ stores the free variables in $A$ and $\Xi$ stores the
assignments of type parameters to variances.
This judgment expresses that if $A$ appears in a context of variance
$\xi$ then all the type variables $\beta \ofv \zeta$ in $\Xi$ will
occur only as prescribed by $\zeta$.  For example, if $\xi = \conv$
and $\beta \ofv \conv$ then all occurrences of $\beta$ in $A$ must
be in covariant or nonvariant positions.  If $\xi = \nonv$ then no
requirement is imposed on any occurrences. We elide the subscript 
$\Sigma_0$ in the rules presented below.

Internal and external choices are covariant in the continuation 
types, \emph{lolli} is contravariant in the first component and 
covariant in the second component, and tensor is covariant in 
both components.
\[
  \begin{array}{c}
 	\infer[\ichoiceop]
    {\vars\semi \Xi \vdash \ichoice{\ell : A_\ell}_{\ell \in L} \ofv \xi}
    {\vars\semi \Xi \vdash A_\ell \ofv \xi \qquad (\forall \ell \in L)}
    \qquad
    \infer[\echoiceop]
    {\vars\semi\Xi \vdash \echoice{\ell : A_\ell}_{\ell \in L} \ofv \xi}
    {\vars\semi\Xi \vdash A_\ell \ofv \xi \qquad (\forall \ell \in L)}
      \end{array}
\]
\[
  \begin{array}{c}
    \infer[\lolli]
    {\vars\semi\Xi \vdash A_1 \lolli A_2 \ofv \xi}
    {\vars\semi\Xi \vdash A_1 \ofv \lnot \xi
    & \vars\semi\Xi \vdash A_2 \ofv \xi}
    \qquad
    \infer[\tensor]
    {\vars\semi\Xi \vdash A_1 \tensor A_2 \ofv \xi}
    {\vars\semi\Xi \vdash A_1 \ofv \xi
    & \vars\semi\Xi \vdash A_2 \ofv \xi}  
  \end{array}
\]

Explicitly quantified types are covariant in the continuation
type. Explicitly quantified variables do not carry any variance
information, so they are kept in the set of variables $\vars$.
\[
    \infer[\exists^x]
    {\vars\semi\Xi \vdash \texists{x} A \ofv \xi}
    {\vars, x\semi\Xi \vdash A \ofv \xi } 
	\qquad
    \infer[\forall^x]
    {\vars\semi\Xi \vdash \tforall{x} A \ofv \xi}
    {\vars, x\semi\Xi \vdash A \ofv \xi } 
\]

The terminated session $\one$ and quantified variables $x$ are 
valid in a context with any variance, whereas type parameters
$\alpha \ofv \xi'\in \Xi$ are valid in contexts of variance at
most $\xi'$. For instance, $\cdot\semi \alpha\ofv\biv \vdash \alpha \ofv \cov$.
\[
    \infer[\one]
    {\vars\semi\Xi \vdash \one \ofv \xi}
    {} 
    \qquad
    \infer[\m{var}]
    {\vars\semi\Xi \vdash x \ofv \xi}
    {x \in \vars} 
    \qquad
    \infer[\m{par}]
    {\vars\semi\Xi \vdash \alpha \ofv \xi}
    {\alpha \ofv \xi' \in \Xi & \xi \leq \xi'}
\]

We instantiate a definition $V[\Xi_V] = A_V$ by writing $V[\theta]$,
where $\theta$ is a \emph{valid substitution} for type parameters in $\Xi_V$.
Each type in $\theta$ must be valid according to the variance 
in the context of $V[\theta]$ and by $\Xi_V$.  We apply the
nesting operator to combine these variances.
\[
      \infer[\m{def}]
      {\vars\semi\Xi \vdash V[\theta] \ofv \xi}
      {V[\Xi_V] = A_V \in \Sigma_0
    & \vars \semi \Xi \vdash \theta \ofv (\Xi_V \mid \xi)}
\]

\paragraph{Valid substitutions}

The judgment for valid substitutions
is written as 
$\vars \semi \Xi \vdash_{\Sigma_0} \theta \ofv \Xi$ and
is defined over a valid signature $\Sg_0$, whose
reference we omit. Again, $\vars$ is a set of variables and $\Xi$ is a set of
type parameters and their corresponding variances.
This judgment expresses that a substitution $\theta$ is valid on a set of variances
$\Xi$ if any type variable $\alpha$ in $\Xi$ is substituted for a type $A$ with
(at least) the same variance.

\[
  \begin{array}{c}
    \infer[(\cdot)]
    {\vars \semi \Xi \vdash (\cdot) \ofv (\cdot)}
    {\mathstrut}
    \qquad
    \infer[\m{subs}]
    {\vars \semi \Xi \vdash (\theta, A/\alpha) \ofv (\Xi_\theta, \alpha \ofv \xi)}
    {\vars \semi \Xi \vdash \theta \ofv \Xi_{\theta}
    & \vars \semi \Xi \vdash A \ofv \xi}
  \end{array}
\]

\paragraph{Properties of variance}
We now identify some properties of variance on nested
polymorphic session types and revisit some examples. 
Additional properties can be found in the 
supplementary material.

\begin{lemma} 
\label{lem:properties_variance_valid_types}
The following properties hold.
\begin{enumerate}
\item If $\vars\semi \Xi \vdash A \ofv \xi$ and $\xi' \leq \xi$ then $\vars\semi \Xi \vdash A \ofv \xi'$.
\item If $\vars\semi \Xi \vdash A \ofv \xi$ and $\vars\semi \Xi, \alpha \ofv \xi \vdash C \ofv \zeta$
  then $\vars\semi \Xi \vdash C[A/\alpha] \ofv \zeta$.
\item If $\vars\semi \Xi \vdash \theta \ofv \Xi'$ and $\vars\semi\Xi' \vdash C \ofv \zeta$
then $\vars\semi \Xi \vdash C[\theta] \ofv \zeta$. 
\item If $\vars\semi \Xi \vdash A \ofv \xi$ then $\vars\semi\Xi \mid \xi' \vdash A \ofv \xi \mid \xi'$.
\item If $\vars\semi \Xi \vdash V[\theta'] \ofv \xi$ and $V[\Xi'] = A' \in \Sigma_0$
then $\vars\semi \Xi \vdash A'[\theta'] \ofv \xi$.
\end{enumerate}
\end{lemma}
\begin{proof}
(1) Proof by induction, using transitivity of $\leq$ and monotonicity of nesting.
(2) Proof by induction, using (1).
(3) Proof by induction, using (2).
(4) Proof by induction, using associativity of nesting.
(5) Proof by inversion, using (4), (3), and $\cov \mid \xi = \xi$.
\end{proof}

\begin{example}
\label{ex:seg_valid}
Consider a signature composed of definitions for 
the type of \emph{list segments} 
presented in~\autoref{sec:overview}:
\[
\begin{array}{llll}
	\Sg_0 &=& \{&\mathsf{List}[\alpha\ofv\cov] = \ichoice{\mathbf{nil} : \one, \mathbf{cons} : \alpha \tensor \mathsf{List}[\alpha]},\\
	&&&\mathsf{Seg}[\alpha\ofv \biv] = \mathsf{List}[\alpha] \lolli \mathsf{List}[\alpha]\,\}
\end{array}
\]
To prove that $\Sg_0$ is valid, we initiate the validity check of $\mathsf{List}[\alpha\ofv \cov]$ with
$\cdot \semi \alpha\ofv\cov \vdash_{\Sigma_0} \ichoice{\mathbf{nil} : \one, \mathbf{cons} : \alpha \tensor \mathsf{List}[\alpha]} \ofv \cov$.
We then explore each branch using the $\ichoiceop$ rule.
Since the $\mb{nil}$ branch is trivial, we show the $\mb{cons}$ branch.
Here, we check $\cdot \semi \alpha\ofv\cov \vdash_{\Sigma_0} \alpha \tensor \mathsf{List}[\alpha] \ofv \cov$,
which, in turn, checks $\cdot \semi \alpha\ofv\cov \vdash_{\Sigma_0} \alpha \ofv \cov$
and $\cdot \semi \alpha\ofv\cov \vdash_{\Sigma_0} \mathsf{List}[\alpha] \ofv \cov$.
The former follows from $\m{par}$ rule, while the latter reduces
to checking $\cdot \semi \alpha\ofv\cov \vdash_{\Sigma_0} \alpha \ofv (\cov \mid \cov)$, and then $\m{par}$ rule applies.

To verify validity of $\mathsf{Seg}[\alpha \ofv \biv]$, we check
$\cdot \semi \alpha\ofv\biv \vdash_{\Sigma_0} \mathsf{List}[\alpha] \lolli \mathsf{List}[\alpha] \ofv \cov$.
Intuitively, this should follow since $\alpha$ has variance $\biv$ in the
context, and therefore can occur covariantly and contravariantly.
Formally, this reduces (by the $\lolli$ rule) to checking $\cdot \semi \alpha\ofv\biv \vdash_{\Sigma_0} \mathsf{List}[\alpha] \ofv \cov$ and $\cdot \semi \alpha\ofv\biv \vdash_{\Sigma_0} \mathsf{List}[\alpha] \ofv \conv$,
which follows by the $\m{def}$ rule and by property (1) of Lemma~\ref{lem:properties_variance_valid_types}.
\end{example}

\section{Subtyping}
\label{sec:subtyping}

The ability to verify subtyping is
paramount to enable the specification of more expressive (well-typed) programs.
However, the combination of recursion, polymorphism, and subtyping
has proven to be challenging.
The subtyping relation for session types was proposed 
by Gay and Hole~\cite{Gay05acta} to enhance the flexibility of the type system.
We say that $A$ is a subtype of $B$, written $A \leq B$, if every behavior 
permitted by $A$ is also permitted by $B$. However, also for 
session types, the combination of 
subtyping and explicit polymorphic quantifiers has proven difficult~\cite{Gay08}.
In this section, we provide the definition of subtyping for nested
polymorphic session types and build on the 
legacy of the undecidability of the inclusion problem for simple
languages~\cite{Friedman76tcs} to prove that subtyping is also undecidable.

Our results show that even for a very simple fragment
of nested polymorphic session types, subtyping is 
undecidable. This fragment is composed of types defined
\emph{mutual recursively} through \emph{unary type definitions}
involving non-deterministic (labelled) choices. 
In our session-typed setting, that means \emph{nested} session types defined through \emph{unary} type definitions whose
type parameter is always assigned positive variance; these definitions only
provide \emph{internal choices}. 

\subsection{Subtyping definition}

We start by defining the notion of \emph{unfolding} on nested polymorphic session types.
We define $\unfold{A}$ recursively on the structure
  of $A$, according to the following rules:
  \begin{mathpar}
    \infer[\m{def}]
    {\unfold{V [\theta]} = B[\theta]}
    {V[\Xi'] = B \in \Sg_0 
    & \vars\semi\Xi \vdash \theta \ofv \Xi'} 
    \and
    \infer[\m{str}]
    {\unfold{A} = A}
    {}
  \end{mathpar}

\paragraph{Simulation and variance-based type relations.} The subtyping
relation for closed types is defined based on the notion of \emph{type simulation}.
Let $\mi{Type}$ denote the set of closed types.

\begin{definition}\label{def:subtp_rel}
  A relation $\rel \subseteq
  \mi{Type} \times \mi{Type}$ is a type simulation if $(A, B) \in
  \rel$ implies the following conditions:
  \begin{itemize}
    \item If $\unfold{A} = \ichoice{\ell : A_\ell}_{\ell \in L}$, then $\unfold{B} =
    \ichoice{m : B_m}_{m \in M}$ where $L \subseteq M$ and $(A_\ell, B_\ell) \in \rel$ for
    all $\ell \in L$.

    \item If $\unfold{A} = \echoice{\ell : A_\ell}_{\ell \in L}$, then $\unfold{B} =
    \echoice{m : B_m}_{m \in M}$ where $L \supseteq M$ and $(A_m, B_m) \in \rel$ for
    all $m \in M$.

    \item If $\unfold{A} = A_1 \lolli A_2$, then $\unfold{B} =
    B_1 \lolli B_2$ and $(B_1, A_1) \in \rel$ and
    $(A_2, B_2) \in \rel$.

    \item If $\unfold{A} = A_1 \tensor A_2$, then $\unfold{B} =
    B_1 \tensor B_2$ and $(A_1, B_1) \in \rel$ and
    $(A_2, B_2) \in \rel$.

    \item If $\unfold{A} = \one$, then $\unfold{B} = \one$.

    \item If $\unfold{A} = \texists{x} A'$, then $\unfold{B} = \texists{y} B'$
    and for all $C \in \mi{Type}$, $(A'[C/x], B'[C/y]) \in \rel$.

    \item If $\unfold{A} = \tforall{x} A'$, then $\unfold{B} = \tforall{y} B'$
    and for all $C \in \mi{Type}$, $(A'[C/x], B'[C/y]) \in \rel$.
  \end{itemize}
\end{definition}

Relying on the idea that covariance is \emph{the variance}, we 
now define relations based on each kind of variance.

\begin{definition}
\label{def:variance_based_rel}
	Given a relation $\rel \subseteq \mi{Type} \times \mi{Type}$,
	we define variance-based relations as follows:
	\begin{itemize}
		\item The covariant-relation of $\rel$ is $\rel^\cov = \rel$.
		\item The contravariant-relation of $\rel$ is $\rel^\conv = \{(A,B) \mid (B,A)\in \rel\}$.
		\item The bivariant-relation of $\rel$ is $\rel^\biv = \{(A,B) \mid 
			  (A,B) \in \rel \text{ and } (B,A)\in \rel\}$.
		\item The nonvariant-relation of $\rel$ is $\rel^\nonv = \{(A,B) \mid 
			  A,B\in \mi{Type}\}$.
	\end{itemize}
\end{definition}

\paragraph{The subtyping relation.}
The definition of subtyping is based on the notion of
type simulation and takes advantage of variance-based 
relations to ensure the required sensitivity to variances.

\begin{definition}[Subtyping]\label{def:subtp_closed}
  Given closed types $A$ and $B$  s.t.
  $\cdot\semi\cdot \vdash A \ofv \xi $
  and
  $\cdot\semi\cdot \vdash B \ofv \xi $,
  we say that $A$ is a subtype of $B$ at variance $\xi$,
  written $A \leq B \ofv \xi$,
  if there exists a type simulation $\rel$ such that 
  $(A, B)\in \rel^\xi $.
\end{definition}

Since our algorithms have to deal with open types,
we can lift this definition by considering suitably
valid substitution instances as follows.

\begin{definition}[Subtyping of open types]\label{def:subtp}
  Given types $A$ and $B$  s.t.
  $\vars\semi\cdot \vdash A \ofv \xi $
  and
  $\vars\semi\cdot \vdash B \ofv \xi $,
  we say that $A$ is a subtype of $B$ at variance $\xi$,
  written $\forall \vars.\, A \leq B \ofv \xi$,
  if
  there exists a type simulation $\rel$ such that 
  $(A[\sigma], B[\sigma])\in \rel^\xi $
  for all closed substitutions $\sigma$ 
  over $\vars$.
\end{definition}

The definition of subtyping could be extended to open types
defined over a non-empty set $\Xi$, just by 
additionally closing the types with all substitutions
$\theta$ and $\theta'$ such that
$\vdash \theta \ofv \Xi$ and 
     $\vdash \theta' \ofv \Xi$, and
$(C,D)\in \rel^{\xi}$ for each $C/\alpha\in\theta$ and 
     $D/\alpha\in\theta'$ and $\alpha \ofv \xi\in \Xi$.
However, at present, we only need to consider an empty $\Xi$ because
(a) syntactic subtyping is only invoked from the type checker 
for the program, and the program only has
the quantified variables in $\vars$, (b) definitions are unfolded 
during the algorithm so we never need to consider a type
with free type parameters.


%
%

\subsection{Undecidability of subtyping}

In this section, we prove that the 
subtyping of nested polymorphic session types is \emph{undecidable}.
Subtyping is undecidable even for a small fragment of these types.
%

In previous work we observed that the type system for 
\polymorphicsessiontypes\ (without explicit quantifiers) has 
many similarities with deterministic pushdown 
automata~\cite{Das21esop}. Even though we have shown that the type equality problem
is decidable,
it is thus perhaps not surprising that the subtyping problem inherits the famous undecidability of the language
inclusion problem for simple languages~\cite{Friedman76tcs}. 
We now prove that the subtyping problem for nested polymorphic session types
 is undecidable by reducing the language inclusion problem
 for Basic Process Algebra (BPA) processes~\cite{GrooteH94}
 to our problem. 
 
\paragraph{Basic Process Algebras.}
\emph{BPA expressions} are defined by the grammar:
\[
  \begin{array}{lrcl}
    \mbox{BPA Expressions} & p,q & ::= &
    a
    \mid X
    \mid p+q
    \mid p \cdot q
  \end{array}
\]
where $a$ ranges over a set of \emph{atomic actions}, $X$ is a \emph{variable},
$+$ represents \emph{non-deterministic choice}, and $\cdot$
represents \emph{sequential composition}~\cite{ChristensenHS95}. 
Recursive \emph{BPA processes} are 
defined by means of \emph{process equations}
$\Delta = \{ X_i \triangleq p_i \}_{i\in I}$, where $X_i$ are distinct 
variables, one of which is identified as the \emph{root}.
A BPA expression is \emph{guarded} if any variable occurs in the 
scope of an atomic action. A system $\Delta = \{ X_i \triangleq p_i \}_{i\in I}$ 
of process equations is guarded if $p_i$ is a guarded expression, for all $i\in I$.
We consider that \emph{all} process equations are guarded.
The \emph{empty process} $\varepsilon$ is the 
neutral element of sequential composition, but does not occur
in a process definition---it is only considered for the 
operational semantics (for more details, see~\cite{ChristensenHS95}).

The operational semantics of a BPA process is a 
labelled transition relation, defined over a set $\Delta$ of guarded
equations as follows:
  \begin{mathpar}
  	\infer[]
  	{a \xrightarrow{a} \varepsilon}{}
  	\and
    \infer[]
    {p+q \xrightarrow{a} p'}
    {p \xrightarrow{a} p'}
     \and
    \infer
    {p+q \xrightarrow{a} q'}
    {q \xrightarrow{a} q'}
    \and
    \infer
    {p\cdot q \xrightarrow{a} p'\cdot q}
    {p \xrightarrow{a} p'}
    \and
    \infer
    {p\cdot q \xrightarrow{a} q}
    {p \xrightarrow{a} \varepsilon}
    \and
  	\infer
    {X \xrightarrow{a} p' }
    {p \xrightarrow{a} p'\enspace (X \triangleq p \in \Delta)}
  \end{mathpar}
The language accepted by a BPA process $p$ is defined as
$L(p) = \{ \overline{a} \mid p \xrightarrow{\overline{a}} \varepsilon\}$.
A process $p$ is \emph{deterministic} if whenever $p\xrightarrow{a} p'$ and
$p\xrightarrow{a} p''$, then $p' = p''$.
A set of process equations $\Delta$ is \emph{normed} if for all 
variables $X$ in $\Delta$, there is a \emph{trace} $\overline{a}$ s.t. 
$X\xrightarrow{\overline{a}} \varepsilon$. A BPA process $p$ is normed if it is 
defined through a normed set of equations. 


\paragraph{Translation of BPA to \PolymorphicSessionTypes}
In the remainder of this section, we focus on normed and deterministic
processes, defined over a set of guarded process equations. 
Let $\mathcal{P}$ denote the set of such BPA processes.
We present the translation of a BPA process $p_0\in \mathcal{P}$, 
defined over a set of BPA equations $\Delta_0$ with root $X_0$, in three steps:
(1) we propose a general translation of BPA guarded expressions to nested session types,
(2) we convert $\Delta_0$ into a type signature $\Sigma_0$,
and (3) we propose a translation for $p_0$.

A \emph{guarded} BPA expression is translated to a 
nested session type without explicit quantifiers and 
parametrized by a type variable $\alpha$. 
This translation is denoted by $\llparenthesis \cdot \rrparenthesis_\alpha$
and is defined as follows:
  \begin{mathpar}
  	\llparenthesis a \rrparenthesis_\alpha = \ichoice{a: \alpha}
  	\and 
  	\llparenthesis X \rrparenthesis_\alpha = X[\alpha/\alpha]
  	\and
  	\llparenthesis a\cdot p + b \cdot q \rrparenthesis_\alpha = 
  			\ichoice{a: \llparenthesis p \rrparenthesis_\alpha,
  					 b: \llparenthesis q \rrparenthesis_\alpha}
  	\and 
    \llparenthesis p \cdot q \rrparenthesis_\alpha = \llparenthesis p \rrparenthesis_\alpha [\llparenthesis q \rrparenthesis_\alpha/\alpha]
    \and 
    \llparenthesis \varepsilon \rrparenthesis_\alpha = \alpha
  \end{mathpar}
 
Through this translation, BPA expressions are converted 
into session types characterized by
providing internal choices and nesting. Atomic actions are translated
into an internal (single) choice with label $a$. Process variables 
lead to type names that are further defined through type definitions 
(detailed below). Non-deterministic choices occurring in \emph{guarded} 
process expressions are characterized by providing atomic actions 
in the scope of any type variable; for this reason, any non-deterministic 
choice can be written in the form $a\cdot p + b \cdot q$, where $a$ and $b$
are atomic actions and $p$ and $q$ are processes. These
choices are converted into internal choices where each branch is 
labelled by the corresponding atomic action. Sequential composition 
is translated into nested types.
Empty continuations are captured by the type parameter $\alpha$. 
Since we are only considering deterministic processes, the translated types
are well-defined.

The type signature $\Sigma_0$ is composed by the translation of all process 
equations in $\Delta_0$.
\begin{equation*}
  \Sigma_0 = \llparenthesis \Delta_0\rrparenthesis_\alpha = \{ X[\alpha\ofv\cov] = \llparenthesis p \rrparenthesis_\alpha
  \mid X\triangleq p \in \Delta \}
\end{equation*}
All type definitions are parametrized by $\alpha$. The parameter $\alpha$
is assigned variance $\cov$ because it only occurs in covariant positions
in the type constructors used in the translation.

Finally, the translation of process $p_0\in \mathcal{P}$,
defined over $\Delta_0$ with root $X_0$, is the nested session type
$X_p [\one / \alpha]$ defined over signature
$\Sg_0=\llparenthesis \Delta \rrparenthesis_\alpha$.
We denote the translation of $p_0$ by $\llparenthesis p_0 \rrparenthesis$.

\begin{example}
The BPA process $p_0$ with root $X_0$, defined over 
$\Delta_0 = \{ X_0 \triangleq a\ \cdot X_0 \cdot c + b \cdot X_1,\enspace X_1 \triangleq a\}$
is translated to type $\llparenthesis p_0 \rrparenthesis=X_0[\one/\alpha]$, defined over the signature composed
by the definitions:
\begin{equation*}
  \begin{array}{lllllll}
  	X_0[\alpha\ofv \cov] &=& \ichoice{a: \llparenthesis X_0 \cdot c \rrparenthesis_\alpha, 
  	                                b : \llparenthesis X_1 \rrparenthesis_\alpha}
  	                     &=& \ichoice{a: X_0 [\ichoice{c: \alpha}/\alpha], 
  	                                b : X_1[\alpha/\alpha]}\\
  	X_1[\alpha\ofv\cov] & =& \ichoice{a : \llparenthesis \varepsilon \rrparenthesis_\alpha}     
  					    & =& \ichoice{a : \alpha}.   
  \end{array}
\end{equation*}
The parallel between $p_0$'s labelled transitions -- $a$, $b$, $c$ -- and 
$\llparenthesis p_0 \rrparenthesis$'s simulation steps -- $\ichoiceop a, \ichoiceop b,
\ichoiceop c$ -- should now become clear.
\end{example}

%

\paragraph{Undecidability of Subtyping}
We now use the translation above to reduce the language 
inclusion problem for BPA processes to
the subtyping problem of nested polymorphic session types. Groote and Huttel
proved that the former is undecidable~\cite{GrooteH94}, thus our 
subtyping problem is also undecidable.
Indeed, this result tells us more:
the subtyping problem is undecidable for the 
smaller fragment of nested polymorphic session types
that allows mutual (parametrized) type definitions only 
involving internal choices and type nesting.

We start with two auxiliary lemmas.

\begin{lemma}
\label{lem:traces_after_trtansitions}
	Given two deterministic and normed processes $p,q\in \mathcal{P}$, 
	if $L(p)\subseteq L(q)$ and $p \xrightarrow{a} p'$
	and $q\xrightarrow{a}q'$, then $L(p')\subseteq L(q')$.
\end{lemma}
\begin{proof}
	Assume that exists $w\in L(p')$ s.t. $w\not\in L(q')$.
	Obviously, $a\cdot w \in L(p)$. However, $a\cdot w\not \in L(q)$
	because $q$ is deterministic and there is no $q''\neq q'$
	s.t. $q\xrightarrow{a}q''$. Thus, we would have $L(p)\not\subseteq L(q)$.
\end{proof}

\begin{lemma}
\label{lem:generic_form_transl}
	Given a deterministic and normed process $p\in \mathcal{P}$,
	$\unfold{\llparenthesis p \rrparenthesis_\alpha}
	= \ichoice{a : \llparenthesis p_a \rrparenthesis_\alpha}_{a\in L}$
	where $L=\{a \mid p \xrightarrow{a} p_a\}$.
\end{lemma}
\begin{proof}
	Proof by induction on the structure of $p$. The cases for atomic action and
	non-deterministic choice are immediate from the definition of 
	$\llparenthesis \cdot \rrparenthesis_\alpha$.
	
	Case $p = X$ and $X \triangleq q \in \Delta$. 
	In this case, for each $a\in L$,
	by the last labelled transition rule, we have 
	$q\xrightarrow{a} p_a$. Since $q$ is guarded, $q$ is of the form
	$q = \sum_{a\in L} a \cdot p_a$. Thus, 
	$\llparenthesis q \rrparenthesis_\alpha
	= \ichoice{a : \llparenthesis p_a \rrparenthesis_\alpha}_{a\in L}$. 
	On the other hand, the type signature contains the definition 
	$X[\alpha\ofv \cov] = \llparenthesis q \rrparenthesis_\alpha$
	and $\llparenthesis p \rrparenthesis_\alpha = X[\alpha/\alpha]$.
	Hence, 
	$\unfold{\llparenthesis p \rrparenthesis_\alpha}
	= \ichoice{a : \llparenthesis p_a \rrparenthesis_\alpha}_{a\in L}$.
	
	Case $p = p_1\cdot p_2$. In this case, for each $a \in L$, either:
	$p_1\xrightarrow{a}p_a'$ and $p_a = p_a'\cdot p_2$;
	or, $p_1\xrightarrow{a}\varepsilon$ and $p_a = p_2$. 
	In any of this subcases,
		by induction hypothesis, 
		$\unfold{\llparenthesis p_1 \rrparenthesis_\alpha}
		= \ichoice{a : \llparenthesis p_a' \rrparenthesis_\alpha}_{a\in L}$,
		with $p_a' = \varepsilon$ in the second subcase.
		Hence,
		$\unfold{\llparenthesis p \rrparenthesis_\alpha}
		= \ichoice{a : \llparenthesis p_a' \rrparenthesis_\alpha[\llparenthesis p_2 \rrparenthesis_\alpha/\alpha]}_{a\in L}
		= \ichoice{a : \llparenthesis p_a'\cdot p_2 \rrparenthesis_\alpha}_{a\in L} $.
\end{proof}

\begin{theorem}
	\label{thm:language-inclusion-subtyping}
	Given two deterministic and normed BPA processes 
	$p,q \in \mathcal{P}$, $L(p)\subseteq L(q)$ if and only if 
	$\llparenthesis p \rrparenthesis \leq \llparenthesis q \rrparenthesis\ofv \cov$.
\end{theorem}

\begin{proof}
	For the direct implication, assume that $L(p)\subseteq L(q)$ and consider a 
	(covariant-)relation over nested polymorphic session types defined by:
	\begin{equation*}
		\mathcal{R} = \{ (\llparenthesis p_0 \rrparenthesis_\alpha[\one/\alpha],\llparenthesis q_0 \rrparenthesis_\alpha[\one/\alpha])
	 \mid  L(p_0) \subseteq L(q_0) \} \cup \{(\one,\one)\}.
	\end{equation*}
	Note that, by definition of process translation, 
	we have $(\llparenthesis p \rrparenthesis, \llparenthesis q \rrparenthesis)\in \mathcal{R}$.
	To check that $\mathcal{R}$ is a simulation, we just need to 
	verify that the conditions of Definition~\ref{def:subtp_rel} are met.
	For that, let
	$(\llparenthesis p_0 \rrparenthesis_\alpha[\one/\alpha],\llparenthesis q_0 \rrparenthesis_\alpha[\one/\alpha]) \in \mathcal{R}$.
	Since $\mathcal{R}$ is only composed by images of the translation, the only conditions
	of Definition~\ref{def:subtp_rel} that
	we end up verifying are the closure conditions for $\ichoiceop$ and $\one$.
	The closure condition for $\one$ is handled by $(\one,\one)\in \mathcal{R}$,
	provided that $p_0$ and $q_0$ do not represent the empty process.
	For $\ichoiceop$, we proceed by case analysis on $p_0$. Throughout the proof we 
	use $\unfold{\cdot}$ to handle type definitions that arise from the process
	equations and from the application of the last rule for the labelled transitions.
	
	Case $p_0=a$. In this case, $p_0\xrightarrow{a} \varepsilon$. Since $L(p_0)\subseteq L(q_0)$,
	then $q_0\xrightarrow{a} \varepsilon$. By definition,
	$\llparenthesis p_0 \rrparenthesis_\alpha = \ichoice{a: \alpha}$.
	Using Lemma~\ref{lem:generic_form_transl}, we have  
	$\unfold{\llparenthesis q_0 \rrparenthesis_\alpha} = \ichoice{a: \alpha }$.
	To prove that $\mathcal{R}$ is a simulation, we need to conclude that 
	$(\alpha[\one/\alpha], \alpha[\one/\alpha])\in \mathcal{R}$, which is immediate from $\mathcal{R}$'s 
	definition, noting that
	$\alpha[\one/\alpha] = \one$.
	
	Case $p_0 = X$ and $X\triangleq p_1\in \Delta$. Since $L(p_0) \subseteq L(q_0)$,
	for every $a$ s.t. $p_1\xrightarrow{a} p_a$ (i.e., $p_0\xrightarrow{a} p_a$), we
	know that $q_0\xrightarrow{a}q_a$ and, using Lemma~\ref{lem:traces_after_trtansitions}, $L(p_a)\subseteq L(q_a)$.
	By Lemma~\ref{lem:generic_form_transl}, we know that 
	$\unfold{\llparenthesis p_0 \rrparenthesis_\alpha}
	= \ichoice{a : \llparenthesis p_a \rrparenthesis_\alpha}_{a\in L}$
	where $L=\{a \mid p_1 \xrightarrow{a} p_a\}$ and
	$\unfold{\llparenthesis q_0 \rrparenthesis_\alpha}
	= \ichoice{b : \llparenthesis q_b \rrparenthesis_\alpha}_{b\in M}$,
	with $L \subseteq M$.
	Since $L(p_a)\subseteq L(q_a)$, we also have 
	$(\llparenthesis p_a \rrparenthesis_\alpha[\one/\alpha], 
	\llparenthesis q_a \rrparenthesis_\alpha[\one/\alpha])\in \mathcal{R}$,
	 for all $a\in L$.
	
	Case $p_0 = a\cdot p_1 + b\cdot p_2$. Recall that this case is representative for non-deterministic
	choice because all process equations are guarded.
	In this case, $p_0\xrightarrow{a} p_1$ and $p_0\xrightarrow{b} p_2$.
	Hence, $q_0\xrightarrow{a} q_1$ and $q_0\xrightarrow{b} q_2$
	and $L(p_1)\subseteq L(q_1)$ and $L(p_2)\subseteq L(q_2)$.
	Thus, 
	$\llparenthesis p_0 \rrparenthesis_\alpha
	= \ichoice{a : \llparenthesis p_1 \rrparenthesis_\alpha, b : \llparenthesis p_2 \rrparenthesis_\alpha}_{a\in L}$
	and, by Lemma~\ref{lem:generic_form_transl},
	$\unfold{\llparenthesis q_0 \rrparenthesis_\alpha}
	= \ichoice{\ell : \llparenthesis q_\ell \rrparenthesis_\alpha}_{\ell\in L}$,
	where $a,b \in L$ and $q_a = q_1$ and $q_b = q_2$.
	Since $L(p_1)\subseteq L(q_1)$ and $L(p_2)\subseteq L(q_2)$,
	we know that $(\llparenthesis p_1 \rrparenthesis_\alpha[\one/\alpha],
	\llparenthesis q_1 \rrparenthesis_\alpha[\one/\alpha])$,
	$(\llparenthesis p_2 \rrparenthesis_\alpha[\one/\alpha],
	\llparenthesis q_2 \rrparenthesis_\alpha[\one/\alpha])\in \mathcal{R}$.
	
	Case $p_0 = p_1\cdot p_2$. 
	In this case, by Lemma~\ref{lem:generic_form_transl},
	we have 
	$\unfold{\llparenthesis p_0 \rrparenthesis_\alpha}
	= \ichoice{a : \llparenthesis p_a \rrparenthesis_\alpha}_{a\in L}$
	with $L=\{a \mid p \xrightarrow{a} p_a\}$.
	Since $L(p_0) \subseteq L(q_0)$, by Lemma~\ref{lem:generic_form_transl}, we know that 
	$\unfold{\llparenthesis q_0 \rrparenthesis_\alpha}
	= \ichoice{m : \llparenthesis q_m \rrparenthesis_\alpha}_{m\in M}$
	with $L \subseteq M$. Furthermore, using Lemma~\ref{lem:traces_after_trtansitions},
	for each $a\in L$,
	$L(p_a) \subseteq L(q_a)$. Thus,
	$ (\llparenthesis p_a \rrparenthesis_\alpha[\one / \alpha],
	 \llparenthesis q_a \rrparenthesis_\alpha[\one / \alpha])\in \mathcal{R} $.
	
	Reciprocally, assume that $L(p)\not\subseteq L(q)$ and let
	$w\in L(p) $ be such that $w\not\in L(q) $.
	Let $w_0$
	be its greatest prefix that occurs in $\m{L}(q)$. 
	We have $ p \ltrans{w_0} p'$ and $q \ltrans{w_0} q'$.
	We can prove, by induction on the length of $w_0$ 
	and using Lemma~\ref{lem:generic_form_transl},
	that any simulation $\mathcal{R}$ for 
	$(\llparenthesis p \rrparenthesis_\alpha[\one/\alpha],
	\llparenthesis q \rrparenthesis_\alpha[\one/\alpha])$
	is such that $(\llparenthesis p' \rrparenthesis_\alpha[\one/\alpha],
	\llparenthesis q' \rrparenthesis_\alpha[\one/\alpha])\in \mathcal{R}$.
	However, since $w_0$ is the greatest (proper) prefix of $w$ for which 
	$ p \ltrans{w_0} p'$ and $q \ltrans{w_0} q'$, we know that there is a labelled
	transition $a$ for $p'$, $ p' \ltrans{a} p''$, that is not
	applicable to $q'$. Hence, 	
	we would have
	a choice label for $\llparenthesis p' \rrparenthesis_\alpha[\one/\alpha]$
	distinct from those in $\llparenthesis q' \rrparenthesis_\alpha[\one/\alpha]$.
	For that, note that all internal labels of $\llparenthesis q' \rrparenthesis_\alpha[\one/\alpha]$
	are derived from labelled transitions in $q'$.
	Thus, we have $(\llparenthesis p' \rrparenthesis_\alpha[\one/\alpha],
	\llparenthesis q' \rrparenthesis_\alpha[\one/\alpha])\not\in \mathcal{R}$. 
	We conclude that there would be no simulation for 
	$(\llparenthesis p \rrparenthesis_\alpha[\one/\alpha],
	\llparenthesis q \rrparenthesis_\alpha[\one/\alpha])$.
\end{proof}

\begin{theorem}
	\label{thm:subtypig-undecidable}
	Checking $A \leq B \ofv \cov $ is undecidable.
\end{theorem}

\begin{proof}
	\autoref{thm:language-inclusion-subtyping} reduces the language inclusion
	problem for deterministic and normed BPA processes to the 
	subtyping problem of 
	(closed) nested polymorphic session types. Groote and Huttel proved the former is 
	undecidable~\cite{GrooteH94}, thus
	the subtyping problem for nested polymorphic session types is also 
	undecidable.
\end{proof}

These results show that subtyping is already undecidable even for a \emph{small fragment}
of nested polymorphic session types: the fragment composed of 
\emph{nested} session types (mutually) defined through \emph{unary} type definitions whose
type parameter is always assigned positive variance; these definitions only
provide \emph{internal choices}. 
\section{Practical Algorithm for Subtyping}
\label{sec:subtp-alg}

Although the subtyping problem is undecidable, we have designed a coinductive algorithm
for approximating this problem. The algorithm is sound but incomplete.
The undecidability of subtyping precludes us from achieving a complete algorithm. However, we
propose a recourse that enables the programmer to provide a \emph{seed}
and help the algorithm \emph{generalize the coinductive hypothesis}.

Taking inspiration from Gay and Hole~\cite{Gay05acta}, we attempt to construct a type simulation.
Our algorithm can terminate in three states: \emph{(i)} types are proved to
have a subtyping relation by constructing a simulation, \emph{(ii)} a counterexample
is detected by identifying a position where the subtype and the supertype differ,
or \emph{(iii)} no conclusive answer is obtained due to algorithm's incompleteness.
We interpret both \emph{(ii)} and \emph{(iii)} as a failure of subtyping
verification, but the extension presented in~\autoref{subsec:eqtypes}
comes to the rescue in the case of \emph{(iii)}.

Our subtyping algorithm is deterministic (with no backtracking) and is presented 
in~\autoref{subsec:subtp_algorithm}. This algorithm is proved
to be sound (\autoref{subsec:soundness_subtp}).
Our algorithm assumes a preliminary pass over the given types to 
introduce \emph{fresh internal names}.

\subsection{Internal renaming}
\label{subsec:internal_renaming}

The fundamental operation in the subtyping algorithm
of our recursive structural types is \emph{loop detection}, 
where we determine if we have already added 
a subtyping relation $A \leq B$ to the type simulation.
However, our simulation also contains open types with free 
variables and, therefore, determining if we have already considered
two types in the subtyping relation becomes a difficult operation.
Following our previous approach~\cite{Das21esop}, we have reduced
this problem to the verification of loop detection on defined 
type names. For this purpose, we perform a renaming of the
given types by introducing fresh internal type names and definitions.

\paragraph{Renaming}
A preliminary type transformation assigns a \emph{fresh} name to
each intermediate (structural) type expression in the given types. 
The new internal names are parametrized over their type variables,
and their definitions are added to the signature.
After the \emph{internal renaming}, 
the type grammar becomes:
\[
\begin{array}{rcl}
  A & ::= & \ichoice{\ell : T}_{\ell \in L}
  \mid \echoice{\ell : T}_{\ell \in I}
  \mid T \lolli T \mid T \tensor T \mid \one \mid \texists{x} T \mid \tforall{x} T \mid x \mid \alpha \\
  T & ::= & V[\theta]
\end{array}
\]
Note that the substitutions $\theta$ are also \emph{internally renamed}, 
implying that the continuation types are a nesting of type names.
In the resulting signature, type names and structural types alternate,
thus allowing loop detection to be entirely performed on defined type names.

\begin{example}
After creating internal names for a list of natural numbers, $\mathsf{List}[\mathsf{nat}]$
where $\mathsf{List}[\alpha] = \ichoice{\mathbf{nil} : \one, \mathbf{cons} : \alpha \tensor \mathsf{List}[\alpha]}$
and $\mathsf{nat} = \ichoice{\mathbf{z} : \one, \mathbf{s} : \mathsf{nat}}$, we
obtain the following declarations:
\begin{mathpar}
	\mathsf{nat} = \ichoice{\mathbf{z} : X_1, \mathbf{s} : \mathsf{nat}}
	\and \mathsf{List}[\alpha] = \ichoice{\mathbf{nil} : X_2, \mathbf{cons} : X_3[\alpha]}\\
	X_1 = \one 
	\and X_2 = \one
	\and X_3[\alpha] = X_4[\alpha] \tensor \mathsf{List}[\alpha]
	\and X_4[\alpha] = \alpha
\end{mathpar}
(To ease the notation, when the type constructors are \emph{unary}, 
we often omit the explicit substitution for the parameter. In $\mathsf{List}$,
for instance, a substitution by $\mathsf{nat}$ should be considered as 
$\theta_{\mathsf{nat}} = (\mathsf{nat}/\alpha)$.)

To illustrate the invariant that continuation types are a nesting
of type names, notice that: a list 
$\mathsf{List}[\ichoice{\mathbf{s} : \ichoice{\mathbf{z} : \one}}]$
is renamed to $\mathsf{List}[X_5]$ under a signature
extended with definitions
\begin{mathpar}
	X_5 = \ichoice{\mathbf{s} : X_6}
	\and X_6 = \ichoice{\mathbf{z} : X_7}
	\and X_7 = \one,
\end{mathpar}
whereas a list of lists of natural numbers, 
$\mathsf{List}[\mathsf{List}[\mathsf{nat}]]$, is already 
renamed, assuming the previous renaming for lists of natural numbers. 
For the latter, the continuation types are now $X_2$ (for branch $\mathbf{nil}$)
and $X_3[\mathsf{List}[\mathsf{nat}]]$ (for branch $\mathbf{cons}$) 
-- both type names.
By unfolding them, we get the structural types
$\one$ and $X_4 \tensor \mathsf{List}[\mathsf{List}[\mathsf{nat}]]$.
The former type does not have any continuation, but the latter
has two continuations: $X_4$ and 
$\mathsf{List}[\mathsf{List}[\mathsf{nat}]]$ -- both type names.
By unfolding, we get structural types again, and so on. 
Our subtyping algorithm takes advantage of this alternation.

%
%
\end{example}

\paragraph{Variance assignment}
Variance assignment is done by inspecting the signature.
We start with declarations $V[\overline{\alpha}] = A$ to elaborate
into declarations $V[\Xi] = A$, assigning a variance $\xi$ to each
type parameter $\alpha$ in each declaration. 

Variances are computed using the variants construction method of Altidor et al.~\cite{Altidor11pldi}: 
we start with the
approximation $\nonv$ for each variance and calculate the least fixed
point to get the most informative variances. Since we have a finite lattice
of variances, this procedure terminates. If a type parameter occurs in 
both covariant and contravariant positions, it is assigned variance $\biv$.
In the degenerate case where $\alpha$ does not occur 
in $A$, we end up with $\alpha\ofv\nonv\in \Xi$.

\begin{example}
\label{ex:signature_list_nat}
After variance assignment, the signature for our list of natural numbers becomes:
\begin{mathpar}
	\mathsf{nat} = \ichoice{\mathbf{z} : X_1, \mathbf{s} : \mathsf{nat}}
	\and \mathsf{List}[\alpha\ofv \cov] = \ichoice{\mathbf{nil} : X_2, \mathbf{cons} : X_3[\alpha]}\\
	X_1 = \one 
	\and X_2 = \one
	\and X_3[\alpha\ofv \cov] = X_4[\alpha] \tensor \mathsf{List}[\alpha]
	\and X_4[\alpha \ofv \cov] = \alpha
\end{mathpar}
\end{example}

\paragraph{Subtyping preservation}
Our internal renaming preserves subtyping. 
To prove this, recall that our subtyping relation relies on the notion of simulation.
Since type simulations are defined over \emph{unfolded} types,  
the notion of simulation is naturally preserved by the internal renaming.

\begin{lemma}
\label{lem:type_sim}
Given a type simulation $\rel \subseteq \mi{Type} \times \mi{Type}$,
if $(A, B) \in \rel$ and $A$ and $B$ are \emph{internally renamed}, then the
following conditions hold:
  \begin{itemize}
    \item If $\unfold{A} = \ichoice{\ell : V_\ell[\theta_{A,\ell}]}_{\ell \in L}$, then $\unfold{B} =
    \ichoice{m : U_m[\theta_{B,m}]}_{m \in M}$ for $L \subseteq M$ and 
    $(V_\ell[\theta_{A,\ell}], U_\ell[\theta_{B,\ell}]) \in \rel$ for all $\ell \in L$.

    \item If $\unfold{A} = \echoice{\ell : V_\ell[\theta_{A,\ell}]}_{\ell \in L}$, then $\unfold{B} =
    \echoice{m : U_m[\theta_{B,m}]}_{m \in M}$ for $L \supseteq M$ and 
    $(V_m[\theta_{A,m}], U_m[\theta_{B,m}]) \in \rel$ for
    all $m \in M$.

    \item If $\unfold{A} = V_1[\theta_{A,1}] \lolli V_2[\theta_{A,2}]$, then $\unfold{B} =
     U_1[\theta_{B,1}] \lolli U_2[\theta_{B,2}]$ and we have  $(U_1[\theta_{B,1}], V_1[\theta_{A,1}]) \in \rel$ and
    $(V_2[\theta_{A,2}], U_2[\theta_{B,2}]) \in \rel$.

    \item If $\unfold{A} = V_1[\theta_{A,1}] \tensor V_2[\theta_{A,2}]$, then $\unfold{B} =
    U_1[\theta_{B,1}] \tensor U_2[\theta_{B,2}]$ and we have  $(V_1[\theta_{A,1}], U_1[\theta_{B,1}]) \in \rel$ and
    $(V_2[\theta_{A,2}], U_2[\theta_{B,2}]) \in \rel$.

    \item If $\unfold{A} = \one$, then $\unfold{B} = \one$.

    \item If $\unfold{A} = \texists{x} V[\theta_{A}]$, then $\unfold{B} = \texists{y} U[\theta_{B}]$
    and for all  $C \in \mi{Type}$, we have $(V[\theta_{A}][C/x], U[\theta_{B}][C/y]) \in \rel$.

    \item If $\unfold{A} = \tforall{x} V[\theta_{A}]$, then $\unfold{B} = \tforall{y} U[\theta_{B}]$
    and for all $C \in \mi{Type}$, we have $(V[\theta_{A}][C/x], U[\theta_{B}][C/y]) \in \rel$.
  \end{itemize}
\end{lemma}

\begin{proof}
	Proof by case analysis on the structure of types after internal renaming
	and applying Definition~\ref{def:subtp_rel}.
\end{proof}

\subsection{Subtyping algorithm}
\label{subsec:subtp_algorithm}

The subtyping algorithm capitalizes on the invariants established
by the internal renaming and, thus, only needs to compare 
two structural types or two type names.
The judgment 
is written as $\vars \semi \G \vdash_{\Sg_0} A \leq B \ofv \delta$
and is defined over a valid signature $\Sg_0$.
$\vars$ stores the free variables in $A$ and $B$
and $\G$ stores the subtyping constraints that are
collected while following the algorithm recursively. 

Under constraints $\Gamma$, we 
attempt to construct a simulation. The subtyping
constraints are of the form
$\clo{\vars}{V_1[\theta_1] \leq V_2[\theta_2] \ofv \delta}$
and are called \emph{closures}.
If a derivation can be constructed, all \emph{closed instances}
of all closures are included in the resulting simulation (see the
proof of~\autoref{thm:soundness}). A closed instance of closure 
$\clo{\vars}{V_1[\theta_1] \leq V_2[\theta_2] \ofv \delta}$ is obtained
by applying a closed substitution $\sigma$ over variables in $\vars$ in a way that 
$V_1[\theta_1][\sigma]$ and $V_2[\theta_2][\sigma]$ have no free type
variables. Note that, since type names are not defined over free quantified 
variables (recall the rules for valid 
signatures,~\autoref{subsec:signatures_types_subst}), 
we can write $V_1[\theta_1][\sigma]$ and $V_2[\theta_2][\sigma]$
as $V_1[\theta_1[\sigma]]$ and $V_2[\theta_2[\sigma]]$.

We present the rules in the remainder of this subsection. 
Since the signature is fixed, we elide it from the rules.
The algorithm is initiated with an empty $\G$ and 
follows the rules of the judgment. 
In the rules below, $T$ and $U$ denote arbitrary 
(internally renamed) types, while $V$ denotes a type name.

Rules for $\oplus$ and $\with$ are specialized to variance. Although 
the continuation types preserve the variance of the context, not surprisingly,
the inclusion of the choice labels is sensitive to the variance. Covariance 
preserves the usual inclusion for subtyping, contravariance reverts the inclusion, 
bivariance requires both inclusions and, thus, the sets of 
labels are equal.
\begin{mathpar}
\infer[\hspace*{-1mm}\oplus_{\cov}]
  {\vars \semi \G \vdash
  \ichoice{l : T_\ell}_{\ell \in L} \leq \ichoice{m : U_m}_{m \in M} \ofv \cov}
  {L \subseteq M & \vars \semi \G \vdash
  T_\ell \leq U_\ell \ofv \cov \quad (\forall \ell \in L)}
  \\
  \infer[\hspace*{-1mm}\with_{\cov}]
  {\vars \semi \G \vdash
  \echoice{\ell : T_\ell}_{\ell \in L} \leq \echoice{m : U_m}_{m \in M} \ofv \cov}
  {L \supseteq M & \vars \semi \G \vdash
  T_m \leq U_m \ofv \cov \quad (\forall m \in M)}
  \\
      \infer[\hspace*{-1mm}\oplus_{\conv}]
  {\vars \semi \G \vdash
  \ichoice{\ell : T_\ell}_{\ell \in L} \leq \ichoice{m : U_m}_{m \in M} \ofv \conv}
  {L \supseteq M & \vars \semi \G \vdash
  T_m \leq U_m \ofv \conv \quad (\forall m \in M)}
  \\
   \infer[\hspace*{-1mm}\with_\conv]
  {\vars \semi \G \vdash
  \echoice{\ell : T_\ell}_{\ell \in L} \leq \echoice{m : U_m}_{m \in M} \ofv \conv}
  {L \subseteq M & \vars \semi \G \vdash
  T_\ell \leq U_\ell \ofv \conv \quad (\forall \ell \in L)}
  \\  
  \infer[\hspace*{-1mm}\oplus_{\biv}]
  {\vars \semi \G \vdash
  \ichoice{\ell : T_\ell}_{\ell \in L} \leq \ichoice{m : U_m}_{m \in M} \ofv \biv}
  {L = M & \vars \semi \G \vdash
  T_\ell \leq U_\ell \ofv \biv \quad (\forall \ell \in L)}	
  \\
  \infer[\hspace*{-1mm}\with_{\biv}]
  {\vars \semi \G \vdash
  \echoice{\ell : T_\ell}_{\ell \in L} \leq \echoice{m : U_m}_{m \in M} \ofv \biv}
  {L = M & \vars \semi \G \vdash
  T_\ell \leq U_\ell \ofv \biv \quad (\forall \ell \in L)}
\end{mathpar}

Tensor preserves the context variances for both components,
whereas \emph{lolli} considers the contravariant position of the
first component and \emph{negates} the variance of its context,
while maintaining the variance for the second component.
\begin{mathpar}
  \infer[\hspace*{-1mm}\tensor]
  {\vars \semi \G \vdash
  T_1 \tensor T_2 \leq U_1 \tensor U_2 \ofv \delta}
  {\vars \semi \G \vdash T_1 \leq U_1 \ofv \delta \qquad
  \vars \semi \G \vdash T_2 \leq U_2 \ofv \delta}
  \\
  \infer[\hspace*{-1mm}\lolli]
  {\vars \semi \G \vdash
  T_1 \lolli T_2 \leq U_1 \lolli U_2 \ofv \delta}
  {\vars \semi \G \vdash T_1 \leq U_1 \ofv \lnot\delta \qquad
  \vars \semi \G \vdash T_2 \leq U_2 \ofv \delta}
\end{mathpar}

The type $\one$ is only a subtype of itself.
\begin{mathpar}
  \infer[\one]
  {\vars \semi \G \vdash \one \leq \one \ofv \delta}
  {}
\end{mathpar}

Explicit quantifiers preserve the variance for their continuations.
Since explicitly quantified variables are only substituted in the
program, they are stored in $\vars$, renamed with a 
\emph{fresh} (quantified) variable $z$.
\begin{mathpar}
  \infer[\exists^z]
  {\vars \semi \G \vdash
  \texists{x}{T} \leq \texists{y}{U} \ofv \delta}
  {\vars, z \semi \G \vdash
  T[z/x] \leq U[z/y] \ofv \delta}
  \and
  \infer[\forall^z]
  {\vars \semi \G \vdash
  \tforall{x}{T} \leq_\delta \tforall{y}{U}}
  {\vars, z \semi \G \vdash
  T[z/x] \leq_\delta U[z/y]}
\end{mathpar}

For quantified variables, we only relate a variable to itself, in a
context of any variance.
\begin{mathpar}
  \infer[\m{var}]
  {\vars \semi \G \vdash
  x \leq x \ofv \delta}
  {}
\end{mathpar}

So far, we have covered the subtyping of structural types.
The rules for type operators compare the components 
according to the variances of their context.
If the type constructors do not match, or the label sets
do not respect the inclusions (for $\ichoiceop$ and $\echoiceop$),
the subtyping fails having constructed a counterexample to 
simulation. Similarly, two type variables are in a subtyping 
relation if and only if they have the same name, as 
exemplified by the $\m{var}$ rule.

We have one special rule for non-variance: all types are 
related in a context with variance $\bot$.
\begin{mathpar}
  \infer[\nonv]
  {\vars \semi \G \vdash
  T \leq U \ofv \nonv}
  {}
\end{mathpar}

Taking advantage of the invariants established by 
the internal renaming, we can now focus on the subtyping 
of two type names.
When comparing type names, we analyze the following cases:
either we are comparing the 
same type name and, taking advantage of it, we focus on 
their substitutions ($\m{refl}$ rule), 
or we already came across a subtyping relation involving
the same type names and we instantiate the quantified variables, 
preserving the relation ($\m{def}$ rule), or 
we expand their definitions
($\m{expd}$ rule).

In the $\m{expd}$ rule, we expand the definitions of $V_1[\Xi_1]$
and $V_2[\Xi_2]$, keeping the variances of the context and adding 
the closure  $\clo{\vars}{V_1[\theta_1] \leq V_2[\theta_2] \ofv \delta}$
to $\Gamma$. Since the subtyping relation between $V_1[\theta_1]$
and $V_2[\theta_2]$, must hold for all its closed instances,
the extension of $\G$ with the corresponding closure enables to 
remember that.
\begin{mathpar}
  \inferrule*[right = $\m{expd}$]
  {V_1[\Xi_1] = A \in \Sg_0 \quad\,
  V_2[\Xi_2] = B \in \Sg_0 \\
  \vars \semi \G, 
  \clo{\vars}{V_1[\theta_1] \leq V_2[\theta_2] \ofv \delta}
  \vdash A[\theta_1] \leq B[\theta_2] \ofv \delta}
  {\vars \semi \G \vdash
  V_1[\theta_1] \leq V_2 [\theta_2] \ofv \delta}
\end{mathpar}

The $\m{refl}$ rule takes advantage of being comparing the same
type name, and calls the subtyping algorithm on the substitutions. 
While doing so, the algorithm updates
the variance of the context by \emph{nesting} the current variance with
the variances of the type parameters. This rule is where 
the nesting of types \emph{harmonizes} with the nesting of variances.
The judgment for the subtyping
of substitutions is (mutual recursively) defined below. 
\begin{mathpar}
  \inferrule*[right=$\m{refl}$]
  {V[\Xi] = A \in \Sg_0 \and
  \vars \semi \G \vdash \theta_1 \leq \theta_2 \ofv \Xi \mid \delta}
  {\vars \semi \G \vdash
  V[\theta_1] \leq V[\theta_2] \ofv \delta}
\end{mathpar}

The $\m{def}$ rule only applies when there already exists a closure
in $\G$ with the same type names $V_1$ and $V_2$, over a set of 
quantified variables $\vars'$.
In that case, we try to find
a substitution $\sigma$ from $\vars'$ to the type expressions
defined over variables in $\vars$ -- that we denote by 
$\mathcal{T}(\vars)$ -- and such that $V_1 [\theta_1]$ and
$V_1 [\theta_1'[\sigma]]$ preserve the subtyping relation, and 
$V_2 [\theta_2'[\sigma]]$ and $V_2 [\theta_2]$ also preserve the 
subtyping relation, under the same variance.
The substitution $\sigma$ is computed by a standard match algorithm
on first-order terms (which is linear-time), applied to the 
syntactic structure of the types. The existence of such substitution ensures
that any closed instance of 
$\clo{\vars}{V_1[\theta_1] \leq V_2[\theta_2] \ofv \delta}$
is also a closed instance of 
$\clo{\vars'}{V_1[\theta_1'] \leq V_2[\theta_2'] \ofv \delta}$,
and these are already present in the constructed type simulation.
So, we can terminate our subtyping check, having successfully 
detected a loop.
\begin{mathpar}
  \inferrule*[right=$\m{def}$]
  { \clo{\vars'}{V_1[\theta_1'] \leq V_2[\theta_2'] \ofv \delta} \in \G \\
  \exists \sigma : \vars' \to \mathcal{T}(\vars) \text{ s.t.} \\
  \left(
  \vars \semi \G \vdash V_1[\theta_1] \leq V_1[\theta_1'[\sigma]] \ofv \delta \wedge
  \vars \semi \G \vdash V_2[\theta_2'[\sigma]] \leq V_2[\theta_2] \ofv \delta\right)}
  {\vars \semi \G \vdash V_1[\theta_1] \leq V_2[\theta_2] \ofv \delta}
\end{mathpar}

\paragraph{Subtyping of substitutions}
The subtyping relation is easily extended to substitutions.
This judgement
is used in the $\m{refl}$ rule above and is
written as $\vars \semi \G \vdash \theta_1 \leq \theta_2 \ofv \Xi$.
Again, the judgement is defined over a valid signature $\Sg_0$ (that is elided),
a set of free variables $\vars$ with variables
occurring in the types that compose $\theta_1$ and
$\theta_2$, and a collection of subtyping closures $\G$.
The two rules for the subtyping of substitutions
ensure that two substitutions are in a subtyping 
relation on a context of variance $\Xi$ if
the types substituting each type parameter 
$\alpha \ofv \delta\in \Xi$ preserve the subtyping relation
under variance $\delta$. 
\begin{mathpar}
  \infer[(\cdot)]
  {
    \vars \semi \G \vdash (\cdot) \leq (\cdot) \ofv (\cdot)
  }
  {}
  \and
  \infer[\m{subs}]
  {\vars \semi \G \vdash (\theta_1, A/\alpha) \leq (\theta_2, B/\alpha) \ofv (\Xi, \alpha \ofv \delta)}
  {\vars \semi \G \vdash \theta_1 \leq \theta_2 \ofv \Xi
  \and
  \vars \semi \G \vdash A \leq B \ofv \delta}
\end{mathpar}



\begin{lemma}\label{lem:alt}
  In a goal $\vars \semi \G \vdash A \leq B \ofv \delta$,
  either $A$ and $B$ are both structural or both type names.
\end{lemma}

\begin{proof}
  By induction on the algorithmic subtyping rules, using the fact that
  type definitions are contractive and that after internal renaming 
  every continuation becomes by a type name.
\end{proof}

\begin{example}
To check that $\mathsf{List}[\mathsf{nat}] \lolli \mathsf{List}[\mathsf{nat}]$
is a subtype of $\mathsf{List}[\mathsf{even}] \lolli \mathsf{List}[\mathsf{nat}]$
in a context of variance $\cov$,
formally written as 
\begin{equation*}
	\cdot \semi \cdot \vdash \mathsf{List}[\theta_{\mathsf{nat}}] \lolli \mathsf{List}[\theta_{\mathsf{nat}}] 
	\leq \mathsf{List}[\theta_{\mathsf{even}}] \lolli \mathsf{List}[\theta_{\mathsf{nat}}]\ofv \cov
\end{equation*}
where $\theta_{\mathsf{nat}} = (\mathsf{nat}/\alpha)$ and 
$\theta_{\mathsf{even}} = (\mathsf{even}/\alpha)$,
we use the signature of Example~\ref{ex:signature_list_nat}
extended with definitions for $\mathsf{even}$ and $\mathsf{odd}$.
Using the subtyping rule for $\lolli$ we need to check that
$\cdot \semi \cdot \vdash \mathsf{List}[\theta_{\mathsf{nat}}] 
 \leq \mathsf{List}[\theta_{\mathsf{even}}] \ofv \conv$
and
$\cdot \semi \cdot \vdash \mathsf{List}[\theta_{\mathsf{nat}}] \leq \mathsf{List}[\theta_{\mathsf{nat}}] \ofv \cov$.
Focusing on the former and using the $\m{expd}$ rule we have
\begin{equation*}
\begin{array}{lll}
	\cdot \semi 
\clo{\cdot}{\mathsf{List}[\theta_{\mathsf{nat}}]  \leq \mathsf{List}[\theta_{\mathsf{even}}]  \ofv \conv}
 \vdash \\
 \hspace*{2cm}\ichoice{\mathbf{nil} : X_2, \mathbf{cons} : X_3[\theta_{\mathsf{nat}}]} \leq
 \ichoice{\mathbf{nil} : X_2, \mathbf{cons} : X_3[\theta_{\mathsf{even}}]} \ofv \conv
	
\end{array}
\end{equation*}
which then compares the continuations. For branch $\mathbf{cons}$ we get
\begin{equation*}
	\cdot \semi 
\clo{\cdot}{\mathsf{List}[\theta_{\mathsf{nat}}]  \leq \mathsf{List}[\theta_{\mathsf{even}}]  \ofv \conv}
 \vdash  X_3[\theta_{\mathsf{nat}}] \leq
 X_3[\theta_{\mathsf{even}}] \ofv \conv
\end{equation*}
to which we now apply the $\m{refl}$ rule to get
\begin{equation*}
	\cdot \semi 
\clo{\cdot}{\mathsf{List}[\theta_{\mathsf{nat}}]  \leq \mathsf{List}[\theta_{\mathsf{even}}]  \ofv \conv}
 \vdash  \theta_{\mathsf{nat}} \leq
 \theta_{\mathsf{even}} \ofv (\alpha\ofv \cov) \mid \conv \; .
\end{equation*}
Recalling that $(\alpha\ofv \cov) \mid \conv = (\alpha\ofv \conv)$, 
the rules for subtyping on substitutions reduce us to
\begin{equation*}
	\cdot \semi 
\clo{\cdot}{\mathsf{List}[\theta_{\mathsf{nat}}]  \leq \mathsf{List}[\theta_{\mathsf{even}}]  \ofv \conv}
 \vdash  \mathsf{nat} \leq
 \mathsf{even} \ofv \conv \; .
\end{equation*}
Recalling that $	\mathsf{nat} = \ichoice{\mathbf{z} : X_1, \mathbf{s} : \mathsf{nat}}$
and 
$	\mathsf{even} = \ichoice{\mathbf{z} : X_5, \mathbf{s} : \mathsf{odd}}$
and 
$	\mathsf{odd} = \ichoice{\mathbf{s} : \mathsf{even}}$,
with 
$X_5 = \one$,
we expand the definitions and enrich the context ($\m{expd}$ rule).
Then,
the relations for branch $\mathbf{z}$, eventually follow from the rule for $\one$,
and for the $\mathbf{s}$ branch we proceed with 
\begin{equation*}
	\cdot \semi 
\clo{\cdot}{\mathsf{List}[\theta_{\mathsf{nat}}]  \leq \mathsf{List}[\theta_{\mathsf{even}}]  \ofv \conv},
\clo{\cdot}{\mathsf{nat} \leq \mathsf{even} \ofv \conv }
 \vdash  \mathsf{nat} \leq
 \mathsf{odd} \ofv \conv \; .
\end{equation*}
After expanding, note that, although we have an internal choice,
we are in a context with negative variance and, thus,
the subset relation in the labels is reverted (according to rule $\ichoiceop_{\cov}$).
We proceed with a comparison of the (single) branch $\mathsf{s}$, which enables us to conclude.
\begin{equation*}
\begin{array}{lll}
		\cdot \semi 
\clo{\cdot}{\mathsf{List}[\theta_{\mathsf{nat}}]  \leq \mathsf{List}[\theta_{\mathsf{even}}]  \ofv \conv},
\clo{\cdot}{\mathsf{nat} \leq \mathsf{even} \ofv \conv },
\clo{\cdot}{\mathsf{nat} \leq \mathsf{odd} \ofv \conv}
 \vdash \qquad \\
 \hfill \mathsf{nat} \leq \mathsf{even} \ofv \conv \; .
\end{array}
\end{equation*}
\end{example}

\subsection{Soundness of subtyping}
\label{subsec:soundness_subtp}

We prove the soundness of the subtyping algorithm by constructing
a simulation from a derivation $\vars \semi \G \vdash A \leq B \ofv \delta $
by \emph{(i)} collecting the conclusions of all the sequents,
and \emph{(ii)} forming all closed instances from them.

\begin{definition}
  Given a derivation $\mathcal{D}_0$ of
  $\vars_0 \semi \G_0 \vdash A_0 \leq B_0 \ofv \delta $, we define the set
  $\mathcal{S}(\mathcal{D}_0)$ of closures.  For each sequent
  of the form
  $\vars \semi \G \vdash A \leq B \ofv \delta $ in $\mathcal{D}_0$,
  we include the closure
  $\clo{\vars}{A \leq B \ofv \delta }$ in $\mathcal{S}(\mathcal{D}_0)$.
\end{definition}

\begin{lemma}[Closure Invariants]\label{lem:closure}
  For any valid derivation $\mathcal{D}$ with the set of closures $\SD$,
  \begin{itemize}
    \item If $\clo{\vars}{\ichoice{\ell : T_\ell}_{\ell \in L} \leq
    \ichoice{m : U_m}_{m \in M}\ofv \cov} \in \SD$ from $\oplus_{\cov}$ rule, then
    $L\subseteq M$ and $\clo{\vars}{T_\ell \leq U_\ell \ofv \cov} \in \SD$ for all $\ell \in L$.
    
    \item If $\clo{\vars}{\ichoice{\ell : T_\ell}_{\ell \in L} \leq
    \ichoice{m : U_m}_{m \in M}\ofv \conv} \in \SD$ from $\oplus_{\conv}$ rule, then
    $L\supseteq M$ and $\clo{\vars}{T_m \leq U_m \ofv \conv} \in \SD$ for all $m \in M$.
    
    \item If $\clo{\vars}{\ichoice{\ell : T_\ell}_{\ell \in L} \leq
    \ichoice{m : U_m}_{m \in M}\ofv \biv} \in \SD$ from $\oplus_{\biv}$ rule, then
    $L = M$ and $\clo{\vars}{T_\ell \leq U_\ell \ofv \biv} \in \SD$ for all $\ell \in L$.
    
    \item If $\clo{\vars}{\echoice{\ell : T_\ell}_{\ell \in L} \leq
    \echoice{m : U_m}_{m \in M}\ofv \cov} \in \SD$ from $\with_{\cov}$ rule, then
    $L\supseteq M$ and $\clo{\vars}{T_m \leq U_m \ofv \cov} \in \SD$ for all $m \in M$.
    
    \item If $\clo{\vars}{\echoice{\ell : T_\ell}_{\ell \in L} \leq
    \echoice{m : U_m}_{m \in M}\ofv \conv} \in \SD$ from $\with_{\conv}$ rule, then
    $L\subseteq M$ and $\clo{\vars}{T_\ell \leq U_\ell \ofv \conv} \in \SD$ for all $\ell \in L$.
    
    \item If $\clo{\vars}{\echoice{\ell : T_\ell}_{\ell \in L} \leq
    \echoice{m : U_m}_{m \in M}\ofv \biv} \in \SD$ from $\with_{\biv}$ rule, then
    $L = M$ and $\clo{\vars}{T_\ell \leq U_\ell \ofv \biv} \in \SD$ for all $\ell \in L$.

    \item If $\clo{\vars}{T_1 \tensor T_2 \leq
    U_1 \tensor U_2 \ofv \delta } \in \SD$ from $\tensor$ rule, then
    $\clo{\vars}{T_1 \leq U_1 \ofv \delta} \in \SD$ and
    $\clo{\vars}{T_2 \leq U_2 \ofv \delta} \in \SD$.

    \item If $\clo{\vars}{T_1 \lolli T_2 \leq
    U_1 \lolli U_2\ofv \delta} \in \SD$ from $\lolli$ rule, then
    $\clo{\vars}{T_1 \leq U_1\ofv \neg\delta} \in \SD$ and
    $\clo{\vars}{T_2 \leq U_2\ofv \delta} \in \SD$.
    
    \item If $\clo{\vars}{\texists{x}T \leq
    \texists{y}U\ofv \delta} \in \SD$ from $\exists^z$ rule, then
     $\clo{\vars,z}{T[z/x] \leq
    U[z/y]\ofv \delta} \in \SD$.
    
    \item If $\clo{\vars}{\tforall{x}T \leq
    \tforall{y}U\ofv \delta} \in \SD$ from $\forall^z$ rule, then
     $\clo{\vars,z}{T[z/x] \leq
    U[z/y]\ofv \delta} \in \SD$.
    
    \item If $\clo{\vars}{V_1[\theta_1] \leq V_2[\theta_2]\ofv \delta} \in \SD$
    and $V_1[\Xi_1]=A \in \Sg_0$
    and $V_2[\Xi_2]=B \in \Sg_0$ from $\m{expd}$ rule, then
    $\clo{\vars}{A[\theta_1] \leq B[\theta_2]\ofv \delta} \in \SD$.

    \item If $\clo{\vars}{V[\theta_1] \leq V[\theta_2]\ofv \delta} \in \SD$
    and $V[\Xi]=A \in \Sg_0$
    from $\m{refl}$ rule, then 
    $\clo{\vars}{A \leq B\ofv \xi\mid\delta} \in \SD$
    for all  $A/\alpha\in\theta_1$ and
  $B/\alpha\in\theta_2$ and $\alpha\ofv\xi\in \Xi$.

    \item If $\clo{\vars}{V_1[\theta_1] \leq V_2[\theta_2]\ofv \delta} \in \SD$
    and $\clo{\vars'}{V_1[\theta_1'] \leq V_2[\theta_2']\ofv \delta} \in \SD$
    from $\m{def}$ rule, then exists $\sigma : \vars' \to \mathcal{T}(\vars)$ s.t.
    $\clo{\vars}{V_1[\theta_1] \leq V_1[\theta_1'[\sigma]]\ofv \delta} \in \SD$
    and
    $\clo{\vars}{V_2[\theta_2'[\sigma]]\leq V_2[\theta_2]\ofv \delta} \in \SD$
  \end{itemize}
\end{lemma}
\begin{proof}
	By induction on the subtyping judgement.
\end{proof}

Now, we prove that the subtyping algorithm that we propose is
sound.

\begin{theorem}
\label{thm:soundness}
  For valid types $A$ and $B$ s.t.
  $\vars \semi \cdot \vdash A \ofv \delta$
  and $\vars \semi \cdot \vdash B \ofv \delta$, if
  the subtyping algorithm proves
  $\vars \semi \cdot \vdash A \leq B \ofv \delta$,
  then $\forall \vars.\, A \leq B \ofv \delta$.
\end{theorem}

\begin{proof}
  Given $A_0, B_0$ s.t. 
  $\vars_0 \semi \cdot \vdash A_0 \ofv \delta$
  and $\vars_0 \semi \cdot \vdash B_0 \ofv \delta$ and a 
  derivation $\mathcal{D}_0$ of $\vars_0 \semi \cdot
  \vdash A_0 \leq B_0 \ofv \delta$, construct the set
  of closures $\mathcal{S}(\mathcal{D}_0)$. Define a relation
  $\rel_0$ such that
  \begin{equation*}
  \begin{array}{lll}
  	    \rel_0  =  \{(A[\sigma], B[\sigma]) \mid \clo{\vars}{A \leq B \ofv \cov}
    \in \mathcal{S}(\mathcal{D}_0), \text{ $\sigma$ closed substitution  over $\vars$}\}\\
     \cup  \{(B[\sigma], A[\sigma]) \mid \clo{\vars}{A \leq B \ofv \conv}
    \in \mathcal{S}(\mathcal{D}_0) \text{ and $\sigma$ closed substitution over $\vars$}\}\\
     \cup  \{(A[\sigma], B[\sigma]), (B[\sigma], A[\sigma])
    \mid \clo{\vars}{A \leq B \ofv \biv}
    \in \mathcal{S}(\mathcal{D}_0), \text{ $\sigma$ closed subst. over $\vars$}\}
  \end{array}
  \end{equation*} 
  Recall the variance-based relation defined in Definition~\ref{def:variance_based_rel}
  and note that $\clo{\vars}{A \leq B \ofv \delta}
  \in \mathcal{S}(\mathcal{D}_0)$ iff $(A[\sigma],B[\sigma])\in\rel_0^\delta$,
  for any closed substitution $\sigma$ over $\vars$.
  Then, let $\rel$ be the reflexive and transitive closure of $ \bigcup_{i\geq 0} \rel_i$ 
  where, for $i\geq 1$:
	\begin{equation*}
	\begin{array}{llll}
		\rel_i = \{(V[\theta_1], V[\theta_2]) \mid & V[\Xi] = C \in \Sg_0 \text{ and }
	(\alpha \ofv \xi) \in \Xi \Rightarrow\\
    & (A/\alpha) \in \theta_1 \text{ and }
    (B/\alpha) \in \theta_2 \; \forall (A,B) \in R_{i-1}^{\xi}\}
	\end{array}
	\end{equation*}
%
  We now prove that $\rel$ is a type simulation. Then, our theorem follows because
	$\vars_0 \semi \cdot \vdash A_0 \leq B_0 \ofv \delta$ implies  
  $\clo{\vars_0}{A_0 \leq B_0 \ofv \delta} \in \mathcal{S}(\mathcal{D}_0)$.
  This means that, for any closed substitution $\sigma$ over $\vars_0$, we have
  $(A_0[\sigma], B_0[\sigma]) \in \rel^\delta$, proving to be in the
  conditions of Definition~\ref{def:subtp}.

  Note that extending a relation by its reflexive and transitive closure
  preserves the simulation properties.
  To prove that $\rel$ is a type simulation, we
  consider $(A[\sigma], B[\sigma]) \in \rel^\delta$ where
  $\clo{\vars}{A \leq B \ofv \delta} \in \mathcal{S}(\mathcal{D}_0)$ and
  $\sigma$ is a substitution over $\vars$. We case analyze on the rule in the
  derivation which added the above closure to $\rel$. 
  We detail the most representative cases; some other cases can be found in
  the supplementary material.

  Consider the case where $\ichoiceop_\conv$ rule is applied
  to add $(B[\sigma], A[\sigma]) \in \rel$, i.e.,
  $(A[\sigma], B[\sigma]) \in \rel^\conv$.
  The rule dictates that $A = \ichoice{\ell : A_\ell}_{\ell \in L}$
  and $B = \ichoice{m : B_m}_{m \in M}$. Since
  $\clo{\vars}{A \leq B \ofv \conv} \in \mathcal{S}(\mathcal{D}_0)$,
  by Lemma~\ref{lem:closure}, we know that $L \supseteq M$ and
  $\clo{\vars}{A_m \leq B_m \ofv \conv} \in \mathcal{S}(\mathcal{D}_0)$ for all
  $m \in M$. By definition of $\rel$, we get 
  $(B_m[\sigma], A_m[\sigma]) \in \rel$. Also,
  $A[\sigma] = \ichoice{\ell : A_\ell[\sigma]}_{\ell \in L}$ and
  $B[\sigma] = \ichoice{m : B_m[\sigma]}_{m \in M}$.  Hence, $\rel$ satisfies the
  first condition of Definition~\ref{def:subtp_rel}.
  Cases for $\ichoiceop_\cov$, $\ichoiceop_\biv$, $\echoiceop_\cov$,
  $\echoiceop_\conv$, $\echoiceop_\nonv$ are analogous.


  If the applied rule is $\lolli$, then $A = A_1 \lolli A_2$
  and $B = B_1 \lolli B_2$. Lemma~\ref{lem:closure} implies
  that $\clo{\vars}{A_1 \leq B_1 \ofv \neg \delta} \in \mathcal{S}(\mathcal{D}_0)$
  and $\clo{\vars}{A_2 \leq B_2 \ofv \delta} \in \mathcal{S}(\mathcal{D}_0)$.
  We thus get $(A_1[\sigma],B_1[\sigma]) \in \rel^{\neg\delta}$
  and $(A_2[\sigma], B_2[\sigma]) \in \rel^\delta$.
  From which we conclude that $(B_1[\sigma], A_1[\sigma]) \in \rel^\delta$ 
  and $\rel$ satisfies the third
  closure condition from Definition~\ref{def:subtp_rel}.
  The case of $\tensor$ is analogous. The case for $\one$
  is trivial.

  If the applied rule is $\exists^z$, then $A = \texists{x} A'$
  and $B = \texists{y} B'$ and $(A[\sigma], B[\sigma]) \in \rel^\delta$.
  Lemma~\ref{lem:closure} implies
  that $\clo{\vars,z}{A'[z/x] \leq B'[z/y] \ofv \delta} \in
  \mathcal{S}(\mathcal{D}_0)$.
  Thus, the definition of $\rel$ ensures that $(A'[z/x][\sigma'], B'[z/y][\sigma']) \in \rel^\delta$
  for any closed $\sigma'$ over $\vars, z$.
  Also, $A[\sigma] = \texists{x} A'[\sigma]$ and
  $B[\sigma] = \texists{y} B'[\sigma]$.
  To prove that we have 
  $(A'[\sigma][C/x], B'[\sigma][C/y]) \in \rel^\delta$ for any $C \in \mi{Type}$,
  we set $\sigma' = \sigma, C/z$, and get
  $A'[z/x][\sigma'] = A'[z/x][\sigma,C/z] = A'[\sigma][C/x]$.
  Similarly, $B'[z/x][\sigma'] = B'[\sigma][C/y]$.
  Combining, we get that $(A'[\sigma][C/x], B'[\sigma][C/y]) \in \rel^\delta$.
  The case for $\forall^z$ is analogous.

  If the applied rule is $\m{var}$, then $A = x$ and $B = x$.
  In this case, the relation $\rel$ contains any $(C, C)$ for
  a closed session type $C$.
  Since the next applied rule has to be a structural rule,
  we can use Lemma~\ref{lem:closure} again to obtain the closure conditions
  of a type simulation.
  If the applied rule is $\nonv$, nothing is added to $\rel$;
  thus, the result holds vacuously. 

  When the $\m{expd}$ rule is applied,
  $A = V_1[\theta_1]$ and $B = V_2[\theta_2]$ and $(A[\sigma], B[\sigma]) \in \rel^\delta$ with 
  definitions $V_1[\Xi_1] = C_1$ and $V_2[\Xi_2] = C_2$.
  From Lemma~\ref{lem:closure},
  we have $\clo{\vars}{C_1[\theta_1] \leq
  C_2[\theta_2] \ofv \delta} \in \SDO$.
  Again, the next applied rule is a structural rule,
  so we use Lemma~\ref{lem:closure} to obtain the closure conditions.

  When the applied rule is $\m{def}$,
  we have
  $	\clo{\vars}{A \leq B \ofv \delta} \in \SDO$,
  with
  $A = V_1[\theta_1]$ and 
  $B = U_1[\theta_2]$,
  and definitions 
  	$V_1[\Xi_1] = C$ and 
  	$U_1[\Xi_2] = D$.
 We also have $\clo{\vars'}{V_1[\theta_1'] \leq U_1[\theta_2']\ofv \delta} \in \SDO$.
  The types $\unfold{A}$ and $\unfold{B}$ are structural. 
  Due to the internal renaming, 
  the continuation types for $C$ and $D$ are type names, 
  say 
  $V_2[\theta_{A}]$ and
  $U_2[\theta_{B}]$, respect., occurring in a position of variance $\xi$.
  Recalling that $(A[\sigma], B[\sigma]) \in \rel^\delta$,
  we want to prove that:
  $
  	(V_{2}[\theta_{A}[\theta_1]][\sigma], U_{2}[\theta_{B}[\theta_2]][\sigma])\in \rel^{\delta | \xi}.
  $
  From Lemma~\ref{lem:closure}, we know that 
  exists $\sigma' : \vars' \to \mathcal{T}(\vars)$ s.t.
  $\clo{\vars}{V_1[\theta_1] \leq V_1[\theta_1'[\sigma']] \ofv \delta},
  	\clo{\vars}{U_1[\theta_2'[\sigma']] \leq U_1[\theta_2] \ofv \delta} \in \SDO$.
  Hence, we have
  $
  	(V_1 [\theta_1[\sigma]], V_1 [\theta_1'[\sigma'[\sigma]]]),
  	(U_1 [\theta_2'[\sigma'[\sigma]]],U_1 [\theta_2[\sigma]])\in \rel^\delta.
  $
  By applying successive simulation steps on these types, we eventually 
  reach the type variables occurring in $C$ and $D$. Thus, we know that:
  \begin{equation}
  \label{eq:thetas_in_rel}
  \begin{array}{ll}
  	(A_1[\sigma], A_1'[\sigma'[\sigma]])\in \rel^{\delta\mid \xi_1}
  	\text{ for all $A_1/\alpha\in\theta_1$ and
  $A_1'/\alpha\in\theta_1'$ and $\alpha\ofv\xi_1\in \Xi_1$}\\
   (A_2'[\sigma'[\sigma]], A_2[\sigma])\in \rel^{\delta\mid \xi_2}
   \text{ for all $A_2'/\beta\in \theta_2'$
  and $A_2/\beta\in \theta_2$ and $\beta\ofv\xi_2\in\Xi_2$}
  \end{array}
  \end{equation}
  The internal renaming ensures that
  $\theta_{A}$ and $\theta_{B}$ are a nesting of type names. Hence,
  from~\eqref{eq:thetas_in_rel} and definition of $\mathcal{R}_i$, we conclude that:
  for some $i \geq 1$,
  \begin{equation}
  \label{eq:bisim_step}
  	(V_{2}[\theta_{A}[\theta_1[\sigma]]], V_{2}[\theta_{A}[\theta_1'[\sigma'[\sigma]]]]),
  	(U_{2} [\theta_{B}[\theta_2'[\sigma'[\sigma]]]], U_{2} [\theta_{B}[\theta_2[\sigma]]])\in 
  	\rel_i^{\delta\mid\xi}.
  \end{equation}
  Since we also have 
  $\clo{\vars'}{V_1[\theta_1'] \leq U_1[\theta_2']\ofv \delta} \in \SDO$, 
  we know that the pair
  $(V_1[\theta_1'][\sigma'[\sigma]], U_1[\theta_2'][\sigma'[\sigma]])\in \rel^{\delta}.$
  Proceeding to the continuation types $V_2[\theta_{A}]$ and
  $U_2[\theta_{B}]$ (occurring at variance $\xi$) and using~\eqref{eq:bisim_step}
  and transitivity,
  we conclude that $	(V_{2}[\theta_{A}[\theta_1]][\sigma], 
  U_{2}[\theta_B[\theta_2]][\sigma])\in \rel^{\delta | \xi}$, as we wanted.

  The last two cases concern reflexivity, one that comes directly from
  the closure obtained from applying the $\m{refl}$ rule, and the other
  comes from the relation $\rel_i$.
  We consider the representative case where closure $\clo{\vars}{V[\theta_1]
  \leq V[\theta_2] \ofv \delta} \in \SDO$, i.e., 
  	$(V[\theta_1][\sigma], V[\theta_2][\sigma])\in \rel^\delta$.
  Assume that $V[\Xi_V] = C \in \Sg_0$.
  Lemma~\ref{lem:closure} ensures that 
  	$\clo{\vars}{A \leq B \ofv \xi\mid\delta} \in \SDO$,
  	for all $A/\alpha \in \theta_1$ and $B/\alpha \in \theta_2$
  and $\alpha \ofv \xi \in \Xi_V$.
  By definition of $\rel_0$, 
  we thus have
  \begin{equation}
  	\label{eq:arguments_in_R0}
  	(A[\sigma],B[\sigma])\in \rel_0^{\xi\mid \delta},
  	\text{ for each $A/\alpha \in \theta_1$ and $B/\alpha \in \theta_2$
  and $\alpha \ofv \xi \in \Xi_V$.}
  \end{equation} 
  We proceed with a continuation of $C$.
  Let $V_0[\theta_0]$ denote a continuation occurring in a position 
  of variance $\eta$ on $C$, with definition $V_0[\Xi_0] = D$. 
  We can easily show
  that $\Xi_0\mid \eta \subseteq \Xi_V$.
  To conclude, recall that
  the intermediate renaming ensures that $\theta_0$ is itself a nesting
  of type names $V_1\cdots V_k$, for $k\geq 1$.
  Using~\eqref{eq:arguments_in_R0} and the definition of $\rel_{k+1}$,
  we have
  $(V_0[\theta_0[\theta_1[\sigma]]], V_0[\theta_0[\theta_2[\sigma]]])\in \rel^{\delta\mid \eta}$,
  as we wanted.
  In the subcase where $C = \alpha$ for $\alpha\ofv \xi \in\Xi$,
  the validity of $\Sg_0$ ensures that $\xi = \cov$.
  In this particular case, we have $V [\theta_1] = A$ and
  $V [\theta_2] = B$ and $A/\alpha \in \theta_1$ and $B/\alpha \in \theta_2$.
  From~\eqref{eq:arguments_in_R0}, we conclude that 
  $(A[\sigma], B[\sigma]) \in \rel_0^{\cov\mid\delta}$ and so
  $(A[\sigma], B[\sigma]) \in \rel^{\delta}$.
  Note that $\unfold{V [\theta_1][\sigma]} = A[\sigma]$.
  The next applied rule has to be a structural rule, so
  we can use Lemma~\ref{lem:closure} again to obtain the closure conditions
  of a type simulation.
  
  Thus, $\rel$ is a type simulation.
\end{proof}

Our subtyping algorithm makes a clear distinction between
the two classes of type variables: \emph{type parameters} are 
always substituted by unfolding definitions during algorithm's 
execution, so they do not arise in the subtyping algorithm, whereas
\emph{explicitly quantified variables} arise in the 
program (as parameters to process definitions) or even in the types, so
they are kept in $\vars$.
Quantified variables are only substituted in the program and never 
in the subtyping algorithm. This is the reason why we 
handle a set of quantified variables $\vars$, but not of
type parameters $\Xi$ in our subtyping algorithm.

\subsection{Subtyping declarations}
\label{subsec:eqtypes}

One source of incompleteness in our algorithm is its
inability to \emph{generalize the coinductive hypothesis}.
As an illustration, consider the following two types $D$ and $D'$.
\begin{equation*}
	\begin{array}{rcllcl}
		T[\alpha \ofv \cov] &=& \ichoice{\mb{L} : T[T[\alpha]], \mb{R} : \alpha}
		\qquad &
		D &=& \ichoice{\mb{L} : T[D], \$ : \one}\\
		T'[\beta \ofv \cov] &=& \ichoice{\mb{L} : T'[T'[\beta]], \mb{R} : \beta}
		\qquad &
		D' &=& \ichoice{\mb{L} : T'[D'], \mb{R}: \one, \$ : \one}
	\end{array}
\end{equation*}
To establish $D \leq D'\ofv \cov$, our algorithm explores the $\mb{L}$
branch and checks $T[D] \leq T'[D'] \ofv \cov $.
A corresponding closure $\clo{\cdot}{T[D] \leq T'[D']}$ is
added to $\Gamma$, and our algorithm then checks $T[T[D]] \leq
T'[T'[D']] \ofv \cov$.
This process repeats until it exceeds the depth bound and terminates
with an inconclusive answer.
What the algorithm never realizes is that $T[\theta] \leq T'[\theta] \ofv \cov$
for all substitutions $\theta$ over $\alpha$; it fails to generalize to this hypothesis and
is always inserting closed subtyping constraints to $\G$.

To allow a recourse, we permit the programmer to declare \emph{subtyping 
declarations} easily verified by our algorithm (concrete syntax):
\begin{Verbatim}
  eqtype T[x]  <= T'[x]
\end{Verbatim}
Indeed, here we could even \emph{declare} an equality constraint,
meaning that, for $T$ and $T'$, subtyping holds in both directions:
\begin{Verbatim}
  eqtype T[x] = T'[x]
\end{Verbatim}

Then, we \emph{seed} the $\G$ in the subtyping algorithm with the corresponding
closure from the $\m{eqtype}$ declaration which can then
be used to establish $D \leq D' \ofv \cov$\,:
\[
  \cdot \semi \clo{x}{T[x/\alpha] \leq T'[x/\beta]}, \clo{x}{T'[x/\beta] \leq T[x/\alpha]} \vdash D \leq D' \ofv \cov\,.
\]
Upon exploring the $\mb{L}$ branch, it reduces to
\[
\begin{array}{lll}
	  \cdot \semi \clo{x}{T[x/\alpha] \leq T'[x/\beta]}, \clo{x}{T'[x/\beta] \leq T[x/\alpha]}, 
  \clo{\cdot}{D \leq D'}
  \vdash \qquad\qquad \\
  \hfill T[D] \leq T'[D'] \ofv \cov.
\end{array}
\]
This last inequality holds under substitution $\sigma = D/x $ over
$\clo{x}{T[x/\alpha] \leq T'[x/\alpha]}$, as required
by the $\m{def}$ rule.

In the implementation, we first collect all the \verb`eqtype`
declarations in the program into a global set of closures $\G_0$.
We then verify the \emph{validity} of every \verb`eqtype` declaration, 
by checking
$\vars \semi \G_0 \vdash A \leq B \ofv \cov$ for every pair $(A,B)$
(with free variables $\vars$) in the \verb`eqtype` declarations.
Essentially, this ensures that all subtyping declarations are
valid with respect to each other.
Finally, all subtyping checks are performed under this more
general $\G_0$. The soundness of this approach can be proved
with the following theorem.
\begin{theorem}[Seeded Soundness]\label{thm:tpeq_sound_seeded}
  For a valid set of $\m{eqtype}$ declarations $\G_0$, if
  $\vars \semi \G_0 \vdash A \leq B\ofv \delta$, then
  $\forall \vars.\, A \leq B \ofv \delta$.
\end{theorem}

Our soundness proof can easily be modified to accommodate
this requirement. Intuitively, since $\G_0$ is valid, all
closed instances of $\G_0$ are already proven to be valid 
subtyping relations and consistent with each other.
Thus, all properties of a type simulation are preserved
when all closed instances of $\G_0$ are added to it.

One final note on the rule of reflexivity: if a type name does \emph{not}
depend on its parameter, at the variance assignment stage, the type 
parameter is assigned variance $\nonv$.  
Hence, its type definition is of the form $V[\alpha\ofv \nonv] = A$.
To prove that $V[T] \leq V[U] \ofv \xi$ using the $\m{refl}$ rule, 
we need to prove that $T \leq U\ofv \nonv$,
which holds due to our special (permissive) rule for $\nonv$ 
and regardless of $T$ and $U$.

\section{Implementation}
\label{sec:implementation}

We have implemented a prototype for nested polymorphic session
types with subtyping and integrated it with the open-source Rast system~\cite{Das20fscd}.
Rast (Resource-Aware Session Types) is a programming language
that implements the intuitionistic version of session types~\cite{Caires10concur}
with support for arithmetic refinements~\cite{Das20CONCUR,Das20PPDP},
ergometric~\cite{das2018work} and temporal~\cite{Das18Temporal} types
for complexity analysis, and nested session types~\cite{Das21esop}.

In this work, however, we focus on the subtyping implementation
and how that connects to the Rast language.
The Rast implementation uses a bi-directional type checker~\cite{pierce2000local}
requiring programmers to specify the initial type of each channel
of a process.
The intermediate types are then reconstructed automatically
using the syntax-directed typing system.

\paragraph{Syntax}
For clarity of the examples in~\autoref{sec:examples}, 
we provide the syntax of our programs, omitting details on process 
definitions. A program contains a series of mutually recursive type 
and process declarations and definitions. It is also provided with
auxiliary subtyping declarations.
\begin{lstlisting}
  type V[x1]...[xn] = A
  eqtype V[x1]...[xm] <= U[y1]...[yk]
  decl f[x1]...[xi] : (c1 : A1) ... (cj : Aj) |- (c : B)
\end{lstlisting}
The first line represents a \emph{type definition}, where $V$ is the type name
parameterized by type parameters $x_1, \ldots, x_n$ and $A$ is its definition. 
The second line stands for \emph{subtyping declarations}, that seed the subtyping 
algorithm with additional constraints. The third
line is a \emph{process declaration}, where $f$ is a process name 
parameterized by type variables $x_1, \ldots, x_i$ and
$(c_1 : A_1) \ldots (c_j : A_j)$ are the used channels and corresponding
types; the offered channel is $c$ and has type $B$. A process declaration is
followed by its definition, which we omit because it is not 
completely necessary for understanding the program's communication 
behavior.

Once the program is parsed and its abstract syntax tree is extracted,
we perform a \emph{validity check} on it. This includes checking that
type definitions, process declarations, and process definitions are closed
with respect to the type variables in scope.
To simplify and improve the subtyping
algorithm's efficiency, we also assign internal names to type subexpressions
parameterized over their free variables. These internal names
are not visible to the programmer. Once the internal names
are collected, the most informative variances are assigned to 
all type parameters in all type definitions. 
For that purpose, we start by assigning variance $\nonv$ 
to all type parameters and then we calculate the least fixed point~\cite{Altidor11pldi}.

\paragraph{Type Checking}
After the most informative variances are assigned to all type parameters,
we check the validity of \verb+eqtype+ declarations. This validity
is performed by checking each declaration against each other 
(recall~\autoref{subsec:eqtypes}).
These declarations are then incorporated 
in the set of closures $\G$, where the algorithm will
collect subtyping constraints during its execution.

The type checker explicitly calls the subtyping algorithm in 3 places:
(i) forwarding, (ii) sending/receiving channels, and (iii) process spawning.
Earlier, Rast used type equality for typing these 3 constructs.
We relax type equality in Rast to subtyping in two ways.
If a process $P$ offers channel of type $A$, we also safely allow $P$
to offer $A'$ if $A \leq A'$. Dually, if a process $P$ uses a channel
of type $A$, we allow it to use $A'$ if $A' \leq A$.

When a type equality needs to be verified, either in the provided
declarations or during algorithm's execution, the type checker 
verifies subtyping in both directions.
\section{Examples}
\label{sec:examples}

In the following examples, we focus on types, rather than process code, because the types alone suffice to describe the essential communication behavior.
The complete Rast code for these examples can be found in the supplementary material.

\subsection{Stacks}

As an example of nested types and subtyping, we will present an implementation of stacks of natural numbers, where the stack type tracks the stack's shape using the nested type parameter.
(These stacks could also be made polymorphic in the type of data elements, but here we choose to use monomorphic data so that we may focus on the use of nested types, subtyping, and keep the syntactic overhead down.
At the end of this section, we will give the polymorphic types.)

\subsubsection{Stacks of natural numbers}

At a basic level, stacks of natural numbers can be described by the (lightly nested) type
\begin{verbatim}
type Stack' = &{ push: nat -o Stack' , pop: Option[Stack'] }
\end{verbatim}
where
\begin{verbatim}
type Option[k] = +{ some: nat * k , none: 1 }
\end{verbatim}
Each stack supports push and pop operations with an external choice between \verb`push` and \verb`pop` labels.
If the stack's client chooses \verb`push`, then the subsequent type, \verb`nat -o Stack'`, requires that the client send a \verb`nat`; then the structure recurs at type \verb`Stack'` to continue serving push and pop requests.
If the stack's client instead chooses \verb`pop`, then the subsequent type, \verb`Option[Stack']`, requires the client to branch on whether the stack is nonempty -- whether there is \verb`some` natural number or \verb`none` at all at the top of the stack.

\subsubsection{Type nesting enforces an invariant}

Implicit in this description of how a stack would offer type \verb`Stack'` is a key invariant about the stack's shape: popping from a stack onto which a natural number was just pushed should always yield \verb`some` natural number, never \verb`none` at all.
However, the type \verb`Stack'` cannot enforce this pop-after-push invariant precisely because it does not track the stack's shape -- \verb`Stack'` can be used equally well to type empty stacks as to type stacks containing three natural numbers, for instance.
But by taking advantage of the expressive power provided by nested types, we can enforce the invariant by defining a type \verb`Stack[k]` of \verb`k`-shaped stacks of natural numbers that can enforce the pop-after-push invariant.

We will start by defining two types that describe shapes.
\begin{verbatim}
type Some[k] = +{ some: nat * k }
type None = +{ none: 1 }
\end{verbatim}
The shape \verb`Some[k]` describes a \verb`k` with \verb`some` natural number added on top; the type \verb`None` describes an empty shape.
Notice that both \verb`Some[k]` and \verb`None` are subtypes of \verb`Option[k]`, for all types \verb`k`.


The idea is that \verb`Stack[None]` will type empty stacks -- because they have the shape \verb`None` -- and that \verb`Stack[Some[Stack[None]]]` will type stacks containing one natural number -- because they have \verb`Some` natural number on top of an empty stack
-- and so on for stacks of larger shapes.

More generally, the type \verb`Stack[k]` describes \verb`k`-shaped stacks:
\begin{verbatim}
type Stack[k] = &{ push: nat -o Stack[Some[Stack[k]]] , pop: k }
\end{verbatim}
Once again, each stack supports push and pop operations with an external choice between \verb`push` and \verb`pop` labels.
This time, however, pushing a \verb`nat` onto the stack
leads to type \verb`Stack[Some[Stack[k]]]`,
i.e., a stack with \verb`Some` natural number on top of a \verb`k`-shaped stack.
Equally importantly, popping from a \verb`k`-shaped stack exposes the shape \verb`k`.

Together these two aspects of the type \verb`Stack[k]` serve to enforce the pop-after-push invariant.
Suppose that a client pushes a natural number, say $0$, onto a stack \verb`s` of type \verb`Stack[k]`.
After the push, the stack \verb`s` will have type \verb`Stack[Some[Stack[k]]]`.
If the client then pops from \verb`s`, the type becomes \verb`Some[Stack[k]]`, which is \verb`+{ some: nat * Stack[k] }`.
This means that the client will always receive \verb`some` natural number, never \verb`none` at all because \verb`none` is not part of this type.
And that is how the type \verb`Stack[k]` enforces the pop-after-push invariant.

Given this type constructor, the empty stack can be expressed as a process of type \verb`Stack[None]`:
\begin{verbatim}
decl empty : . |- (s : Stack[None])
\end{verbatim}
And we can define a process \verb`elem[k]` that constructs a stack with shape given by \verb`Some[Stack[k]]`, from a natural number and a stack of shape \verb`k`:
\begin{verbatim}
decl elem[k] : (x:nat) (t : Stack[k]) |- (s : Stack[Some[Stack[k]]])
\end{verbatim}

\subsubsection{Subtyping}

So far, we have taken advantage of the expressive power of nested types to guarantee adherence to an invariant about a stack's shape.
The strength and precision of these types is a double-edged sword, however.
Sometimes we will want to implement operations for which it is difficult or even impossible to maintain the stronger, more precise types across the operation.
Subtyping will allow us to fall back on the weaker, less precise types in these cases.

For example, our more precise stack types \verb`Stack[None]` and \verb`Stack[Some[k]]` are both subtypes of the less precise \verb`Stack'` type if \verb`k` is a subtype of \verb`Stack'`.
So, as a specific example, from stacks of type \verb`Stack[None]` or \verb`Stack[Some[Stack']]`, we can always fall back on the less precise \verb`Stack'`.
Our Rast implementation is capable of verifying this instance
using the following \verb`eqtype` declarations.
\begin{verbatim}
eqtype Stack[None] <= Stack'
eqtype Stack[Some[Stack']] <= Stack'
eqtype Stack[Option[Stack']] <= Stack'
\end{verbatim}



As an example of where this subtyping can be useful, consider an operation to reverse a stack.
Because our system does not include intersection types -- an opportunity for future work -- it is difficult to express the invariant that reversing a stack preserves its shape, making it difficult to use the more precise \verb`Stack[None]` and \verb`Stack[Some[k]]` types.
Instead, we use the type \verb`Stack'` to describe the stack before and after reversal:
\begin{verbatim}
decl reverse : (t : Stack') |- (s : Stack')
\end{verbatim}
This \verb`reverse` process would be implemented by way of a helper process:
\begin{verbatim}
decl rev_append : (t : Stack') (a : Stack') |- (s : Stack')
\end{verbatim}
Specifically, we would call \verb`rev_append` with the accumulator \verb`a` being the \verb`empty` stack process mentioned previously.
But \verb`empty` has the type \verb`Stack[None]`, whereas \verb`rev_append` uses the less precise \verb`Stack'` type.
Subtyping is how we are able to bridge this gap and implement stack reversal.



Above, we showed that we can move from a more precise \verb`Stack[None]` or \verb`Stack[Some[Stack']]` type to a less precise \verb`Stack'` type using subtyping, and showed how that can be useful.
But by moving to the less precise \verb`Stack'` type, we are not irreversibly losing information about the stack's shape.
Given a \verb`Stack'`, we can recover the more precise type \verb`Stack[Option[Stack']]` for that stack by using the process
\begin{verbatim}
decl shaped : (t : Stack') |- (s : Stack[Option[Stack']])
\end{verbatim}
This \verb`shaped` process acts as a kind of coercion that analyzes the shape of \verb`Stack'` and constructs a \verb`Stack[Option[Stack']]`.
This coercion is operationally not the identity function, so \verb`Stack'` is not a \emph{subtype} of \verb`Stack[Option[Stack']]`, but we can still use it as an explicit coercion for recovering information about the stack's shape.


\subsubsection{Polymorphic stacks}

Lastly, we close this example by generalizing the above types for stacks of natural numbers to types for stacks that are polymorphic in the type \verb`a` of data that they hold.

The type for polymorphic stacks of unknown shape is:
\begin{verbatim}
type Stack'[a] = &{push: a -o Stack'[a], pop: Option[a][Stack'[a]]}
\end{verbatim}
and the type for polymorphic stacks of shape \verb`k` is:
\begin{verbatim}
type Stack[a][k] = &{push: a -o Stack[a][Some[a][Stack[a][k]]], pop: k}
\end{verbatim}
where
\begin{verbatim}
type Option[a][k] = +{ some: a * k , none: 1 }
type Some[a][k] = +{ some: a * k }
type None = +{ none: 1 }
\end{verbatim}

Notice that the type \verb`Stack'[a]` is bivariant in \verb`a`, with \verb`a` occurring in both negative and positive positions.
Thus there is no subtyping relationship at all between \verb`Stack'[nat]` and \verb`Stack'[even]`.

The type \verb`Stack[a][k]` is also bivariant in \verb`a`, although for a slightly more intricate reason:
\verb`Stack[a][k]` is covariant in \verb`k` and so the occurrence of \verb`a` in \verb`Some[a][Stack[a][k]]` in the larger type \verb`Stack[a][Some[a][Stack[a][k]]]` is a positive one, in contrast with the negative occurrence in the linear implication.
Like for \verb`Stack'[ ]`, there is no subtyping relationship between types \verb`Stack[nat][k]` and \verb`Stack[even][k]`.

\subsection{Queues}

We can similarly use a combination of nested types and subtyping to describe queues.
The type \verb`Queue'` describes a queue with an external choice between \verb`enq` and \verb`deq` labels:
\begin{verbatim}
type Queue' = &{ enq: nat -o Queue' , deq: Option[Queue'] }
\end{verbatim}
As for stacks, there is a dequeue-after-enqueue invariant that is not enforced by this type.
We can again use nested types to enforce this invariant:
\begin{verbatim}
type Queue[k] = &{ enq: nat -o Queue[Some[Queue[k]]] , deq: k }
\end{verbatim}

At first glance, it seems obvious that this should be possible for queues -- it worked for stacks, after all.
However, the nature of parametric type constructors like \verb`Stack[k]` and \verb`Queue[k]` is that we can only build up types from the parameter \verb`k`, never analyze \verb`k` to break it down into smaller types.
This means that the nested type operates something like a stack.
But unlike stacks, queues insert new elements at the back, so it is, in fact, somewhat surpring that this same technique works to enforce the dequeue-after-enqueue invariant for queues.

Similar to the stacks, we can move from the more precise shaped types \verb`Queue[None]` and \verb`Queue[Some[Queue']]` to the less precise type \verb`Queue'` by way of subtyping.
But, like the stacks, we can also express a coercion from the less precise \verb`Queue'` to the more precise \verb`Queue[Option[Queue']]`:
\begin{verbatim}
decl shaped : (q' : Queue') |- (q : Queue[Option[Queue']])
\end{verbatim}

\subsection{Dyck language}

As a simple example, consider the language of \verb`$`-terminated strings of balanced parentheses (also known as the Dyck language)~\cite{Dyck1882}.
This language can be described with the type \verb`D0` that uses a type constructor \verb`D[k]` to track the number of unmatched \verb`l`s:
\begin{verbatim}
type D0 = +{ l: D[D0] , $: 1 }
type D[k] = +{ l: D[D[k]] , r: k }
\end{verbatim}
Each occurrence of \verb`D[ ]` correspond to an unmatched \verb`l`.

We can also describe the related language \verb`l`$^n$\verb`r`$^n$\verb`$` (with $n \geq 0$) using nested types:
\begin{verbatim}
type E0 = +{ l: E[+{$:1}] , $: 1 }
type E[k] = +{ l: E[R[k]] , r: k }
type R[k] = +{ r: k }
\end{verbatim}
Here the number of type constructors \verb`R[ ]` corresponds to the number of unmatched \verb`l`s.

The \verb`$`-terminated language of balanced parentheses includes language \verb`l`$^n$\verb`r`$^n$\verb`$`, and this is reflected in the subtyping relationship between \verb`E0` and \verb`D0`: \verb`E0` is a subtype of \verb`D0`.
Our Rast implementation will verify this, provided that we sufficiently generalize the coinductive hypothesis using \verb`eqtype` declarations:
\begin{verbatim}
eqtype R[k] <= D[k]
eqtype E[R[k]] <= D[D[k]]
\end{verbatim}

(Strictly speaking, our types are coinductively defined, so it is not quite right to think of this as language inclusion for finite strings, but instead a property of all of the finite prefixes of potentially infinite strings.)

\section{Additional Related Work}
\label{sec:related-work}

The literature on subtyping in general is vast, so we cannot survey
it; some related work is already mentioned in the introduction.  Once
recursive types and parameterized type constructors are considered
together, there are significantly fewer results, as far as we are aware.

The seminal work in this area in general is that of
Solomon~\cite{Solomon78popl} who showed that the problem of language
equivalence for deterministic pushdown automata (DPDA) is 
interreducible with \emph{type equality} for a
language with purely positive parameterized recursive types under an
inductive interpretation.  Our construction for undecidability has
been inspired by his technique, although we use a reduction from the
related problem of language inclusion for BPAs, rather than DPDAs, due to the coinductive
nature of our definition of subtyping.

Unfortunately, despite the eventual discovery that DPDA language
equality is decidable~\cite{Senizergues02tcs} it does not yield a
practical algorithm.  Instead, our algorithm is based on Gay and
Hole's construction of a simulation~\cite{Gay05acta} for session
types, combined with careful considerations of constructor variance
(following~\cite{Altidor11pldi}) and instantiation of type variables 
when (in essence) constructing circular 
proofs~\cite{Brotherston05tableaux} of subtyping.

Perhaps closest to our system in the functional world is that by Im et 
al.\ which features equirecursive nested types with a
mixed inductive and coinductive definition of type equality and
subtyping~\cite{Im13icalp}.  It also includes modules, but not first-class universal or
existential quantification.  The authors show their system to be sound
for a call-by-value language under the assumption of contractiveness
for type definitions (which we also make).  They only briefly mention
algorithmic issues with no resolution.

\section{Conclusion}
\label{sec:conclusion}

We presented an undecidability proof for a system of subtyping with
nested recursive and polymorphic types under a coinductive
interpretation of types.  We also presented a practical algorithm for
subtyping.  When embedded in a language with bidirectional type
checking and explicit polymorphism, type checking directly reduces to
subtyping and we have implemented and explored our algorithm in the
context of Rast~\cite{Das20fscd}.

Since undecidability of a small fragment is definitive, the most
interesting items of future work concern extensions of the algorithm.

It seems plausible that we can hypothesize additional pairs to add to
the partial simulation in situations when the subtyping algorithm
exceeds its depth bound and consequently fails without a
counterexample.  Currently, this must be done by the programmer.

We conjecture it requires only a minor adaptation to incorporate a
mixed inductive/coinductive definition of subtyping into our
algorithm~\cite{Danielsson10mpc,Brandt98fi}.  This would make it
easily applicable in strict functional languages such as ML\@.

Much work has been done on subtyping polymorphic
types~\cite{Tiuryn02iandc} in implicit form so that
$\forall \alpha.  \alpha \rightarrow \alpha \leq \m{int} \rightarrow
\m{int}$.  We do not consider such subtyping here, since it introduces
new sources of undecidability and significant additional algorithmic
complexities.  It is an interesting question to what extent our
\emph{algorithm} might still apply in the setting of implicit
polymorphism, possibly restrict to the prefix case.

A final item for future work is to consider \emph{bounded
  polymorphism} and extension from nested types to the related
\emph{generalized algebraic data types} (GADTs).  A positive sign here
is the uniform semantics by Johann and Polonsky~\cite{Johann19lics},
disturbed slightly by more recent results pointing out some obstacles
to parametricity~\cite{Johann21arxiv}.


\bibliographystyle{splncs04}
\bibliography{refs}


\end{document}